\documentclass[reprint,
superscriptaddress,
 amsmath,amssymb,
 aps, longbibliography
]{revtex4-1}

\usepackage{graphicx}
\usepackage{dcolumn}
\usepackage{bm}
\usepackage{qcircuit}
\usepackage{braket}
\usepackage{soul}
\usepackage{tikz}
\usepackage{amsthm}
\usepackage{xcolor}
\usepackage{hyperref}
\usepackage{multirow}
\usepackage{array}
\usepackage{stmaryrd}
\newtheorem{lemma}{Lemma}
\newtheorem{theorem}{Theorem}
\newtheorem{theorem2}{Theorem}

\newcommand{\rev}[1]{{#1}}

\begin{document}

\title{Low-overhead fault-tolerant quantum computing using long-range connectivity}
\author{Lawrence Z. Cohen}
\affiliation{Centre for Engineered Quantum Systems, School of Physics, University of Sydney, Sydney, New South Wales 2006, Australia}

\author{Isaac H. Kim}
\affiliation{Department of Computer Science, UC Davis, Davis, CA 95616, USA}
\affiliation{Centre for Engineered Quantum Systems, School of Physics, University of Sydney, Sydney, New South Wales 2006, Australia}

\author{Stephen D. Bartlett}
\affiliation{Centre for Engineered Quantum Systems, School of Physics, University of Sydney, Sydney, New South Wales 2006, Australia}

\author{Benjamin J. Brown}
\email{b.brown@sydney.edu.au}
\affiliation{Centre for Engineered Quantum Systems, School of Physics, University of Sydney, Sydney, New South Wales 2006, Australia}

\begin{abstract}
    Vast numbers of qubits will be needed for large-scale quantum computing due to the overheads associated with error correction. We present a scheme for low-overhead fault-tolerant quantum computation based on quantum low-density parity-check (LDPC) codes, where long-range interactions enable many logical qubits to be encoded with a modest number of physical qubits.  In our approach, logic gates operate via logical Pauli measurements that preserve both the protection of the LDPC codes as well as the low overheads in terms of the required number of additional qubits. Compared with surface codes with the same code distance, we estimate order-of-magnitude improvements in the overheads for processing around one hundred logical qubits using this approach. \rev{Given the high thresholds demonstrated by LDPC codes, our estimates suggest that fault-tolerant quantum computation at this scale may be achievable with a few thousand physical qubits at comparable error rates to what is needed for current approaches.} \\
    
\end{abstract}

\maketitle

\section{Introduction}

Quantum computing devices are now capable of outperforming even the fastest conventional supercomputers at certain tasks~\cite{arute2019quantum}.  However, to execute many quantum algorithms of practical interest, it is widely believed that a fault-tolerant architecture will be required to identify and correct errors in noisy quantum hardware.  Fault-tolerant architectures come with a significant overhead cost, using a large number of low-noise physical qubits to encode and process quantum information with even a small number of protected qubits.
Specifically, it has been estimated that millions of qubits will be needed to solve relevant problems in quantum chemistry~\cite{vonBurg2021,Lee2020,Kim21}
, to break cryptosystems~\cite{Fowler12,Gidney21} or to get an advantage over classical algorithms using polynomial speedups~\cite{Campbell19,Sanders2020}
. These large overheads provide a daunting challenge for scaling up from today's noisy devices to large-scale fault-tolerant quantum computers.

The enormous resource estimates mentioned above are all obtained using fault-tolerant architectures based on quantum error-correcting codes with \emph{local} check operators~\cite{Dennis02,Bombin06,Bacon06,BonillaAtaides21}
. These codes have a number of highly desirable features for quantum computation, including high thresholds and fast decoders~\cite{Raussendorf07,Fowler12}.  The locality of these codes means that quantum error correction can proceed using only entangling gates between neighbouring qubits arranged in a two-dimensional layout, i.e., on a chip. Thus, while local codes provide a clear pathway to demonstrate the principles of fault tolerance using existing quantum technology
, these overheads mean that useful fault-tolerant quantum computing with this approach will likely remain out of reach in the near term.

Locality of gate operations is a physically well-motivated constraint.  Recently, however, there has been significant progress in developing long-ranged entangling gate operations in a variety of quantum processing systems, including those based on superconductors~\cite{Magnard2020}, semiconductors~\cite{Mills2019,Borjans2020,yoneda2021coherent} and trapped ions~\cite{debnath2016demonstration,pino2021demonstration}. Optical photons provide an approach that is not naturally constrained to a local two-dimensional layout~\cite{bombin2021interleaving,bartolucci2021fusion}, and can also allow for other qubit systems to be connected into complex quantum networks~\cite{Barrett05,Nickerson13,Kalb17,Stephenson20}. \rev{Recent work has also considered emulating long-range interactions using a local quantum architecture and classical communication~\cite{delfosse2021bounds}, and architectures have been proposed where long-range interactions are constrained on interconnected planar arrays of matter-based qubits~\cite{tremblay2021constantoverhead}.} The possibility of long-range connectivity opens the door to a new class of quantum codes and fault-tolerant architectures that can harness this connectivity to our advantage.

Here, we show how to perform fault-tolerant quantum computation with an architecture that exploits long-range connectivity to significantly reduce the overhead, compared with local approaches.  Rather than focusing on asymptotic behaviour, we consider the overhead savings that may be possible in the scale of devices expected in the near term, where for example fault-tolerant quantum computing on 50 logical qubits may be possible with only a few thousand physical qubits while maintaining a code distance of $d=14\sim 16$. For comparison, a surface-code based architecture requires at least ten-thousand qubits to attain a similar number of logical qubits and code distances.
Provided that long-range coupling becomes sufficiently reliable to go below the fault-tolerance threshold of our scheme, we anticipate that such an architecture will be capable of performing non-trivial quantum algorithms at a scale compatible with current roadmaps for quantum devices under development during the next few years.

\begin{table*}
    \centering
    \begin{tabular}{||c|c|c|m{7em}<{\centering}|m{4em}<{\centering}|m{4em}<{\centering}|m{4em}<{\centering}||}
        \hline
        $k$ & $d$ & Parallelism & Code family & $n_{\rm data}$ & $n_{\rm anc}$ & $n_{\rm tot}$\\
        \hline
        \hline
        \multirow{2}{2em}{\centering $18$} & \multirow{2}{2em}{\centering $8$} & \multirow{2}{*}{$2$} & Hyperbicycle & 294 & 500 & 800 \\
        & & & Surface & 1152 & 128 & 1300 \\
        \hline
        \multirow{4}{*}{$50$} & \multirow{2}{*}{\hfil $14$} & \multirow{2}{*}{$2$} & Hyperbicycle  & 900 & 1400 & 2300 \\
        & & & Surface  & 9800 & 300 & 10000 \\
        \cline{2-7}
        & \multirow{2}{*}{\hfil $16$} & \multirow{2}{*}{$20$} & Hypergraph & 1922 & 5000 & 7000 \\
        & & & Surface & 12800 & 2000 & 15000 \\
        \hline 
        \multirow{4}{*}{\hfil $578$} & \multirow{4}{*}{\hfil $16$} & \multirow{2}{*}{$578$} & Hypergraph & 7938 & 120000 & 130000 \\
        & & & Surface & 150000 & 75000 & 225000 \\
        \cline{3-7}
        & & \multirow{2}{*}{$68$} & Hypergraph & 7938 & 15000 & 23000 \\
        & & & Surface & 150000 & 10000 & 160000 \\
        \hline
    \end{tabular}
    \caption{\textbf{Overhead estimates.}  Estimates of the overhead required to perform a round of logic, including those qubits needed to encode the data as well as additional ancilla qubits required to perform fault-tolerant gates. We use LDPC codes constructed in~\cite{Panteleev19, Kovalev2013}, which all have initial check weights of no more than 10. We denote the number of logical qubits as $k$ and the distance of the code as $d$. Comparisons are made against the surface code with the same distance.  Here, `parallelism' denotes the number of logical qubits that can be acted upon non-trivially in one round of error correction, and which determines the number of required ancilla qubits.  The number of data, ancillary, and total physical qubits needed to perform one round of logical measurements with error correction are denoted $n_{\rm data}$, $n_{\rm anc}$, and $n_{\rm tot}$, respectively. We do not include any ancilla qubits that may be used for error syndrome extraction. \rev{Estimates for the surface code were obtained using the compact block scheme from Ref.~\cite{Litinski19}.} }
    \label{tab:results}
\end{table*}

Our approach uses quantum low-density parity-check (LDPC) codes, which efficiently encode a large amount of logical information for a given number of physical qubits. There has been a recent surge of interest in this subject, (see Ref.~\cite{Breuckmann21} for a recent review) spurred by Gottesman's remarkable observation~\cite{Gottesman14} that quantum LDPC codes meeting certain criteria can be used to achieve fault-tolerant quantum computing with constant overhead.  While research into quantum LDPC codes is still in its infancy, they are showing promise.  Codes that fulfill Gottesman's criteria are now known~\cite{Fawzi20}. Moreover, recent numerical studies indicate that LDPC codes can achieve reasonably high thresholds~\cite{Panteleev19,Roffe20,Grospellier2021}. Recent breakthroughs in achieving high code distances indicate that there is room for further development~\cite{Hastings20,Panteleev20,Breuckmann2021a}.

To use these LDPC codes for quantum \emph{computation}, one must be able to fault-tolerantly implement a universal set of protected logic gates. While Ref.~\cite{Gottesman14} establishes a method to perform quantum computation using fault-tolerant gate teleportation~\cite{Knill2003}, the cost associated with the distillation of the requisite resource state~\cite{Zheng_2020} is not understood well in the practical regime of interest.

In this paper, we introduce a flexible method to perform low-overhead quantum logic gates for a general class of quantum LDPC codes.  Our work can be thought as a generalization of lattice surgery~\cite{Horsman12}, where an ancillary system is coupled to a quantum error-correcting code to measure logical Pauli operators in a fault-tolerant way. Our approach to low-overhead quantum logic builds on an extensive literature into the use of code deformations to perform Clifford gates via measurement that have been well studied for topological codes~\cite{Raussendorf07,Fowler12,Horsman12,Brown17,Litinski19,Vuillot19}, and which have recently been generalized to certain classes of quantum LDPC codes~\cite{Breuckmann17,Lavasani18,Krishna20,Krishna21}. To employ this approach for quantum LDPC codes in a way that maintains the desirable low overheads, we construct the required ancillary system by adapting weight-reduction methods proposed in Refs.~\cite{Hastings16,Hastings21} to measure the desired logical operators of a given quantum LDPC code. These logical operations then yield a universal gate set for fault-tolerant quantum computing when supplemented with noisy ancilla state injection via magic-state distillation~\cite{Bravyi05}. \rev{Thus, our scheme provides an explicit way of performing low-overhead fault-tolerant quantum computing using quantum LDPC codes that is applicable to codes of even modest size.}


Before presenting our detailed results, we briefly illustrate the potential overhead improvements that our construction enables.  We will make use of existing quantum LDPC codes with explicit constructions and efficient decoders, together with our fault-tolerant approach to performing logic gates on these codes. Table 1 shows the overhead required to complete a round of logical operations with error correction for a given number of logical qubits $k$ and code distance $d$ for a number of quantum LDPC codes, specifically, hyperbicycle and hypergraph product codes explicitly constructed in Refs.~\cite{Panteleev19,Kovalev2013}. \rev{By one round of logical operations we mean a set of logical Pauli measurements such that each logical qubit is acted on non-trivially by at most one measurement.}
We directly compare the qubit resources for our construction against the use of surface codes encoding the same number of logical qubits and with the same code distance, with the latter serving as a proxy for how well the codes protect logical quantum information. The surface code is currently the predominant candidate for a quantum architecture, and considerable effort has been spent optimizing its overhead for fault-tolerant computation. Our analysis thus shows the potential overhead improvement that can be achieved using a non-local architecture as compared to a local architecture.

All codes we have used, and their fault tolerant operations, use check operators involving no more than 13 \rev{qubits for the largest code in Table 1. (However, it is not the case that all stabilizer generators will have this weight, and in fact most generators will have lower weight.)} This weight is larger than that of the surface code, but not by a significant margin. Because this number is fixed at a small constant value, errors do not spread significantly during measurement of these check operators, which would otherwise affect the threshold of the scheme. \rev{However, we emphasize that the stabilizer weights and the circuits used to measure the stabilizer generators will still affect the performance of this scheme. In particular, higher weight stabilizers will increase the failure rate under circuit-level noise and this should be taken into account when considering the estimates in Table 1}. Efficient decoders for \rev{the codes in Table 1} have also been designed that perform comparably to minimum-weight perfect matching decoding on surface codes of similar distance~\cite{Panteleev19}.

We implement Clifford gates through parity measurement of logical qubits in the Pauli basis. In order to keep the overhead low, we restrict the number of logical qubits that can participate in a single measurement round, and we call this number the \emph{parallelism} of the scheme; see Fig. 1. For a given level of parallelism and given error correcting code we require $n_{\rm anc}$ physical qubits to create the ancilla systems used in logical measurement, and $n_{\rm data}$ is the size of the code used to store the logical information. Our analysis shows that at very small code sizes, quantum LDPC codes give only a modest overhead improvement when compared to surface code architectures. However as the size of the system increases we see that the improvement in overhead for the quantum LDPC codes becomes very significant, offering an order of magnitude improvement at several hundred logical qubits. Current algorithms for minimal non-trivial quantum chemistry calculations \cite{Babbush2018} require around $100{-}200$ logical qubits and so even at this regime it may be beneficial to use quantum LDPC codes. We expect that for larger algorithms, such as more complex quantum chemistry algorithms~\cite{vonBurg2021,Lee2020,Kim21} or Shor's algorithm~\cite{Gidney21,Gouzien21}, where several thousand logical qubits are required the overhead gains will be substantial.

\begin{figure}
    \centering
    \includegraphics[width=0.45\textwidth]{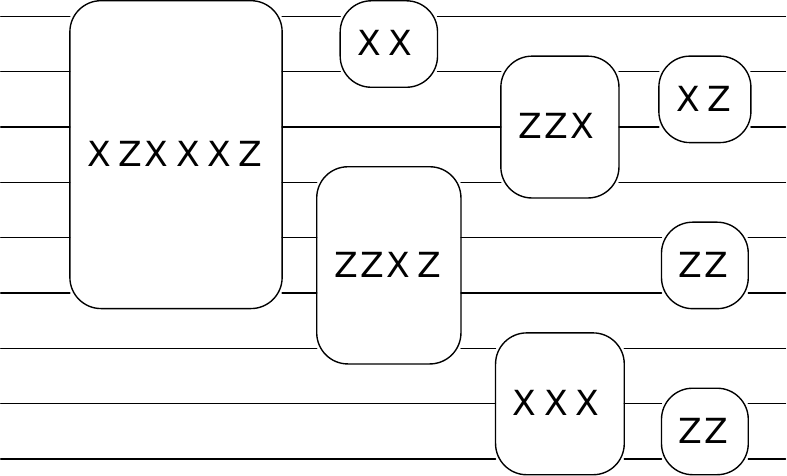}
    \caption{\textbf{Parallelism.}  An example circuit consisting of Pauli measurements on an architecture with a parallelism of $6$. This circuit contains $4$ rounds of error corrected logical measurements.  In each round, at most $6$ logical qubits in total can be involved in logical measurements.}
    \label{fig:paralellism}
\end{figure}

These encouraging reductions in overhead motivate experimental work towards the design and realization of quantum LDPC codes in the laboratory.  At the physical level it remains to find efficient ways to measure check operators extending over distant qubits with high fidelity. With long-range coupling now demonstrated using a number of very different approaches~\cite{Magnard2020,Mills2019,Borjans2020,yoneda2021coherent,debnath2016demonstration,pino2021demonstration}, as well as proposals allowing all-to-all coupling~\cite{ramette2021}, there is significant room for innovation here. It is also critical that we identify check operator readout circuits that maintain the fault tolerance of the scheme, \rev{as well as account for cross-talk that may be present in long-range interactions}, and this will require the development of quantum LDPC codes at the level of the circuit error model~\cite{Higgott2021,tremblay2021constantoverhead}.

Lastly, further study into compilation of quantum algorithms will allow us to determine what level of parallelism is required to efficiently execute a quantum computation.  A commonly employed elementary fault-tolerant gate set consists of non-destructive measurements of arbitrary Pauli strings~\cite{Litinski19,Kim21,chamberland2021universal}. If the parallelism is strictly less than the number of logical qubits, clearly not all the Pauli strings can be measured directly. In particular, if the weight of the Pauli string exceeds the parallelism, one would need to break that measurement down into a sequence of (lower-weight) Pauli measurements, leading to a reduced speed. However, if the majority of the Pauli strings have small weights, the speed of the two approaches will not differ significantly. We anticipate this to be the case if the goal is to simulate a locally interacting many-body quantum spin Hamiltonian using the Trotter-Suzuki method~\cite{Wecker2014}, but it is unclear how the two approaches will differ for methods such as qubitization~\cite{Low2019}.

\section{\label{sec:measurements}Results}

We now present our main result:  a procedure to implement fault-tolerant logic gates in quantum LDPC codes via a generalization of lattice surgery, in a way that preserves the low overheads.  We propose a method for implementing fault-tolerant gates on quantum LDPC codes by using multi-logical Pauli measurements~\cite{Litinski19}. Our result expands on a set of techniques originally devised to reduce the weight of stabilizer generators of quantum codes~\cite{Hastings16,Hastings21}. We extend these results to include measurement of all logical Pauli operators, allowing for implementation of the full logical Clifford group. Our method is a form of code deformation, in which we transform our code into a new code and in doing so obtain logical information about our original code~\cite{Raussendorf07,Fowler12,Horsman12,Brown17,Litinski19,Vuillot19}.

Our construction enables us to perform single-qubit Pauli measurements, as well as parity measurements between logical qubits on one or multiple LDPC blocks in an arbitrary choice of Pauli basis. This capability gives us a measurement-based approach for realising the full Clifford gate set~\cite{Litinski19}. This gate set can be supplemented with magic state distillation and state injection to achieve universal quantum computing~\cite{Bravyi05}. We will discuss universal quantum computing in more detail in Section~\ref{sec:FTQC}. Our construction also guarantees that the distance of the code is preserved during the code deformation and so we retain the error-correcting capabilities of our code.

We begin by setting some basic notation and terminology in Section~\ref{sec:ldpc}.  Section~\ref{sec:deformation} presents our construction for measuring logical multi-qubit Pauli operators. We then prove that our construction preserves the distance of the code throughout the process in Section~\ref{sec:fault_tol}.

\subsection{\label{sec:ldpc}Notation and terminology}

We describe quantum error-correcting codes with the stabilizer formalism
. Let $\mathcal{P} = \langle I, X, Y, Z\rangle $ be the Pauli group and $\mathcal{P}_n = \mathcal{P}^{\otimes n}$ the Pauli group acting on $n$ qubits.  A stabilizer code is defined by an Abelian group $\mathcal{S} \subset \mathcal{P}_n$ such that $-I \not\in \mathcal{S}$. 
The code $\mathcal{C}$ is a subspace spanned by the common $+1$ eigenvalue eigenstates of the operators in $\mathcal{S}$. The logical Pauli operators are operators in $\mathcal{P}_n$ that commute with every operator in $\mathcal{S}$ but are not themselves in $\mathcal{S}$. If $\mathcal{S}$ is generated by an independent set of generators $\{ g_1, \ldots, g_m\}$, the number of logical qubits of the code is $k=n-m$. The distance of the code is equal to the weight of the least-weight non-trivial logical operator where the weight of an operator in $\mathcal{P}_n$ is the number of qubits on which it acts non-trivially, i.e., with non-identity support. Of particular interest are a class of stabilizer codes known as Calderbank-Shor-Steane (CSS) codes. These are codes with a stabilizer group that can be generated by a set that includes only two types of elements: those that are the product of Pauli-X operators only and those the product of only Pauli-Z operators.


We consider families of stabilizer codes $\mathcal{S}_n$ such that each member of family is indexed by the number of qubits $n$.  Let $w_n$ be the maximum weight of a stabilizer of a generating set of $\mathcal{S}_n$ and let $q_n$ be the maximum number of stabilizer generators that act on any given qubit for a specified generating set. A family of LDPC codes then is a sequence $\mathcal{S}_n$ that can be generated by a set of stabilizer generators such that $w_n = O(1)$ and $q_n = O(1)$. For the remainder of the paper we will drop the explicit dependence on $n$, and assume these quantities are constants.

We can describe LDPC codes using the Tanner graph. Let $\mathcal{T} = (V, C, E)$ be a bipartite graph. Each node in $V$ corresponds to a physical qubit in the code, and each node in $C$ corresponds to a generator of the stabilizer group of the code. We draw an edge between $c \in C$ and $v \in V$ if the generator corresponding to $c$ acts non-trivially on the qubit corresponding to $v$. We label each edge with either an $X$, $Z$, or $Y$ depending on how the generator acts on the qubit. For the case of CSS codes we can instead label the nodes in $C$ with $X$ or $Z$, as opposed to labelling the edges. We will often abuse notation and use the labels for variable and check nodes of a Tanner graph to refer to their respective physical qubits and stabilizer generators directly.

\subsection{\label{sec:deformation}Code deformation}

We now present our main technical contribution, starting with an outline of the basic idea behind implementing measurements using code deformation. Suppose we have a stabilizer code $\mathcal{S} = \langle g_1, \ldots, g_m \rangle $ along with a logical operator $\tilde{L}$ that we wish to measure. We can interpret our procedure as adding $\tilde{L}$ to the generating set of a new code that includes this operator, and then removing it again once we have reliably obtained the measurement outcome. Note, however, that simply adding $\tilde{L}$ to $\mathcal{S}$ will \emph{not} yield an LDPC code because in general $\tilde{L}$ is a high weight operator. To maintain the key properties of an LDPC code, we need a stabilizer code that includes $\tilde{L}$ in its stabilizer group whereby $\tilde{L}$ can be generated with constant- and ideally low-weight generators. This can be achieved by creating an extended system with additional qubits that include a set of low-weight stabilizer generators $S_1, \ldots, S_l$ whose product gives $\tilde{L}$. The product of the measurement results for each stabilizer generator gives us the measurement result of $\tilde{L}$. It is also important that the stabilizers of the new code do not generate any other logical operators of $\mathcal{S}$ so that we do not make any unwanted logical measurements that may affect our computation.

\begin{figure*}
    \centering
    \includegraphics[width=0.95\textwidth]{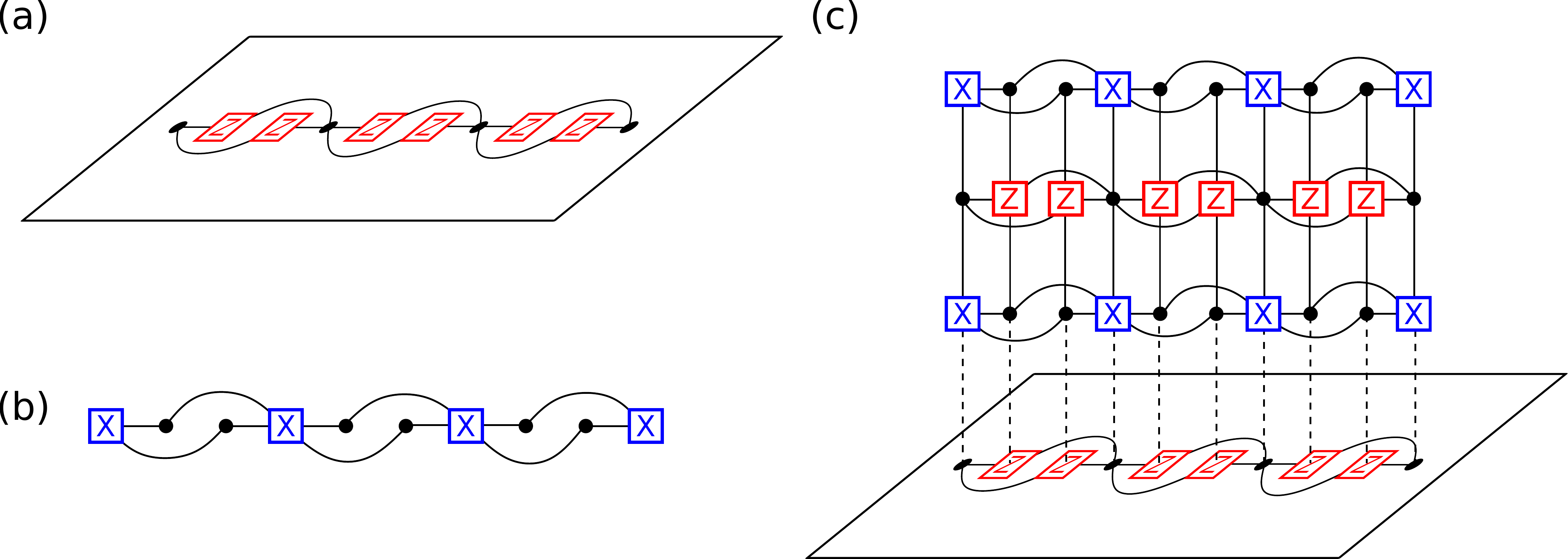}
    \caption{\textbf{Measurement of a logical $\tilde{X}$ operator of the code $\mathcal{C}$.} \textbf{(a)}  Bipartite subgraph $\mathcal{G}_{\tilde{X}}$ of the Tanner graph of $\mathcal{C}$ on the support of $\tilde{X}$.  Black nodes are the variable nodes corresponding to qubits in the support of $\tilde{X}$.  Red nodes are the check nodes corresponding to $Z$-type stabilizers in $\mathcal{C}$ that act on qubits in the support of $\tilde{X}$. \textbf{(b)}  The dual graph $\mathcal{G}^T = (V^T_{\tilde{X}}, C^T_{\tilde{X}}, E^T_{\tilde{X}})$ of the logical $\tilde{X}$ in (a). There is a one-to-one mapping between the $X$-type generators and the qubits in (a), and the qubits and the $Z$-type generators in (a). \textbf{(c)} Measurement of $\tilde{X}$ using the ancilla system $\mathcal{G}_{\textrm{anc}} = \mathcal{G}_{\textrm{merged}} \backslash \mathcal{G}$. The Tanner graph $\mathcal{G}_{\textrm{anc}}$ is constructed by taking alternating layers of the subgraph $\mathcal{G}^{T}_{\tilde{X}}$ in (b) and the subgraph $\mathcal{G}_{\tilde{X}}$ in (a). The vertical edges are the set $E_{\textrm{extra}}$, which connect adjacent layers. The product of the $X$ generators gives the logical $\tilde{X}$ and hence the product of the measurement results for each $X$ generator gives the measurement result of the logical $\tilde{X}$. After merging the codes and measuring $\tilde{X}$ we then split the codes by measuring the stabilizers for $\mathcal{C}$ and measuring the qubits in $\mathcal{C}_{\textrm{anc}}$ in the $Z$ basis, returning us to the original code space.}
    \label{fig:merge}
\end{figure*}

\subsubsection{\label{sec:single}Logical CSS measurement}

Here we describe our construction for fault-tolerant measurement of logical Pauli operators in CSS codes, illustrated with an example in Fig. 2 for concreteness. We concentrate on the case for measuring $X$ logical operators, but remark that an analogous procedure will hold for measuring $Z$ logical operators by reversing the $X$ and $Z$ terms in the following discussion. Let $\mathcal{C}$ be a CSS stabilizer code described by Tanner graph $\mathcal{G} = (V, C, E)$. We show how to measure a specific instance of an $X$ logical operator $\tilde{X}$. To do so, we deform code $\mathcal{C}$ onto a new code $\mathcal{C}_\textrm{merged}$ described by the Tanner graph that includes $\mathcal{G}$ as a subgraph, $\mathcal{G}_\textrm{merged} \supset \mathcal{G}$. 

Our goal is to define a new LDPC code whose stabilizer group includes $\tilde{X}$. We can use this new code to infer the measurement result of $\tilde{X}$ by measuring its stabilizer generators. 
The construction of $\mathcal{G}_{\textrm{merged}}$ is expressed in terms of $\mathcal{G}$ and a subgraph of $\mathcal{G}$ specified by $\tilde{X}$ that we call $\mathcal{G}_{\tilde{X}} = \left(V_{\tilde{X}}, C_{\tilde{X}}, E_{\tilde{X}}\right)$, together with a dual graph $\mathcal{G}_{\tilde{X}}^{T} = (V_{\tilde{X}}^{T}, C_{\tilde{X}}^{T},E_{\tilde{X}}^{T})$. The variable nodes $V_{\tilde{X}}$ of the subgraph $\mathcal{G}_{\tilde{X}}$ are the qubits that support $\tilde{X}$; checks $C_{\tilde{X}}$ are $Z$ type stabilizers \rev{that share an edge of the Tanner graph with} variable nodes $V_{\tilde{X}}$; and edges $E_{\tilde{X}} \subseteq E$ are those of $E$ that are incident to nodes included in both $V_{\tilde{X}}$ and $C_{\tilde{X}}$; see Fig. 2(a).  The dual graph $\mathcal{G}_{\tilde{X}}^T$ is such that for each $v \in V_{\tilde{X}}$ we have a corresponding node $v^{T} \in C^{T}_{\tilde{X}}$ and likewise for each $c \in C_{\tilde{X}}$ we have a corresponding vertex $c^{T} \in V^{T}_{\tilde{X}}$. For each edge $e \in E_{\tilde{X}} $ with $e = (v,c)$ we have $e^{T} \in E^{T}_{\tilde{X}}$ with $e^{T} = (v^{T},c^{T})$ such that the nodes $v^{T} \in C^T_{\tilde{X}}$ and $c^{T} \in V^T_{\tilde{X}}$ are those that correspond to the nodes of the original subgraph; see Fig. 2(b).
For now we assume that there is no strict subset of qubits $V' \subset V$ such that $V'$ supports a distinct $X$ logical operator, before explaining the differences with the more general case.

We define the Tanner graph $\mathcal{G}_{\textrm{merged}}$ in terms of $\mathcal{G}$ and our new graphs $\mathcal{G}_{\tilde{X}}$ and  $\mathcal{G}_{\tilde{X}}^T$.  Specifically, we combine the Tanner graph $\mathcal{G}$ with $r$ copies of $\mathcal{G}^{T}_{\tilde{X}}$ and $r-1$ copies of $\mathcal{G}_{\tilde{X}}$ using additional edges $E_\textrm{extra} \subset E_\textrm{merged}$. We layer the copies of $\mathcal{G}_{\tilde{X}}$ and  $\mathcal{G}_{\tilde{X}}^T$ in an alternating fashion; see Fig. 2(c). The additional edges $E_{\textrm{extra}}$ connect adjacent layers of $\mathcal{G}^{T}_{\tilde{X}}$ and $\mathcal{G}_{\tilde{X}}$ as shown in Fig.~\ref{fig:merge}(c). To explicitly describe the edges of $E_\textrm{extra}$, we index copies of these graphs and their corresponding objects $\mathcal{G}^{T}_{\tilde{X}}[j]$ with $1 \le j \le r$ and $\mathcal{G}_{\tilde{X}}[j]$ with $2 \le j \le r$. We can regard $\mathcal{G}_{\tilde{X}} \subseteq \mathcal{G}$ as $\mathcal{G}_{\tilde{X}}[1]$. In later sections we will refer to the layers corresponding to $\mathcal{G}^T_{\tilde{X}}[j]$ a \emph{dual layers}, and the layers corresponding to $\mathcal{G}_{\tilde{X}}[j]$ as \emph{primal layers}. We will also refer to the final layer $\mathcal{G}^T_{\tilde{X}}[r]$ as the \emph{boundary layer}.

Let us also append indices to the objects of $\mathcal{G}_{\tilde{X}}[j] = (V[j],C[j], E[j])$ and $\mathcal{G}^{T}_{\tilde{X}}[j] = (V^{T}[j],C^{T}[j], E^{T}[j])$ to define $E_{\textrm{extra}}$. Recall that for each variable and check in $\mathcal{G}_{\tilde{X}}$ we have a corresponding check or variable, respectively in $\mathcal{G}^{T}_{\tilde{X}}$. Thus, we also have that each $v[j] \in V[j]$ has a corresponding vertex $v^T[k] \in C^{T}[k]$ and likewise for each $c[j] \in C[j]$ we have $c^T[k] \in V^{T}[k]$. The appending edges $E_{\textrm{extra}}$ then include all edges $( v[j], v^T[j])$ and $(c[j], c^{T}[j])$ for all $v$, $c$ and $1 \le j \le r$, and edges $( v[j+1], v^T[j])$ and $(c[j+1], c^{T}[j])$ for all $v$, $c$ and $ 2 \le j \le r-1 $. 

To summarise, with $E_\textrm{extra} \subset E_\textrm{merged}$ defined, we have now specified all of the objects of $\mathcal{G}_{\textrm{merged}}$. We have that $\mathcal{G}$ is a subgraph of $\mathcal{G}_{\textrm{merged}}$. Likewise, the variables of $\mathcal{G}^{T}_{\tilde{X}}[j]$ and $\mathcal{G}_{\tilde{X}}[j]$ are variables of $\mathcal{G}_{\textrm{merged}}$. The checks of subgraphs $\mathcal{G}_{\tilde{X}}[j]$ are $Z$ stabilizers for $\mathcal{G}_{\textrm{merged}}$ and the checks of dual graphs $\mathcal{G}_{\tilde{X}}^{T}[j]$ are $X$ stabilizers of the merged Tanner graph. Edges $E_\textrm{extra}$ fix graphs $\mathcal{G}$, $\mathcal{G}_{\tilde{X}}[j]$ and $\mathcal{G}^{T}_{\tilde{X}}[j]$ together. It will sometimes be helpful to refer to $\mathcal{G}_\textrm{anc} = \mathcal{G}_\textrm{merged} \backslash \mathcal{G}$ as the ancilla system---the resource used to make parity measurements.

Having constructed our new code defined by the Tanner graph $\mathcal{G}_\textrm{merged}$, we now present the following lemma demonstrating that $\tilde{X}$ is an element of the stabilizer group of this new code.

\begin{lemma} \label{lemma:merge}
Let $\mathcal{C}$ be a CSS LDPC code and $\tilde{X}$ an $X$ logical operator in $\mathcal{C}$ such there is no other $X$ logical operator supported on a strict subset of the qubits in $\tilde{X}$. 
Then the construction for $\mathcal{G}_\textrm{merged}$ above gives a code $\mathcal{C}_{\textrm{merged}}$ that contains $\tilde{X}$ in its stabilizer group whereby
\begin{equation}
    \prod_{\substack{ j = 1,  \ldots , r \\ v^T[j]\in C^T[j]}} v^T[j] = \tilde{X}. 
\end{equation}
\end{lemma}
\noindent
We recall that we have used the notation where $v_T[j]$ is the stabilizer generator represented by this check node.

\begin{proof}
We prove this lemma by showing that the support of $\prod_{ j, C^T[j]} v^T[j] $ on the qubits of $\mathcal{G}_\textrm{anc}$ is trivial, and is non-trivial on $V[1]$, which is the support of $\tilde{X}$.  By definition, stabilizers $v^T[j]$ are supported on variables $v[j] \in V[j]$ and $c^T[j] \in V^T[j]$. We therefore concentrate on these qubits.

Observe that each $Z$ generator $c \in C[1]$ must be connected to an even number of qubits in $V[1]$ in order for the stabilizers to commute with $\tilde{X}$. Consequently, each physical qubit $c^T[k] \in V^T[k]$ must be connected to an even number $q_X$ of $X$ generators in $C^T[k]$. Hence each qubit $c^T[k]$ supports the term $(X_{c^T[k]})^{q_x} = \openone$ for the operator $\prod_{ j, C^T[j]} v^T[j] $.

Furthermore, each physical qubit $v[k] \in V[k]$ for $k \geq 2$ is connected to exactly two $X$ generators, $v^T[k-1]$ and $v^T[k]$. It follows that each qubit $ v[k] \in V[k] $ supports the term $(X_{c^T[j]})^{2} = \openone$ for the operator $\prod_{ j, C^T[j]} v^T[j] $. 

Finally, each physical qubit $v[1] \in V[1]$ is connected to exactly one $X$ generator in $\mathcal{G}_{\textrm{anc}}$, specifically $v^T[1]$. Consequently, the product of all the $X$ generators in $\mathcal{C}_\textrm{anc}$ gives precisely the logical operator $\tilde{X}$, thus giving us the desired result.
\end{proof}

To make a measurement of $\tilde{X}$ in a practical way we first determine $\mathcal{G}_\textrm{merged}$ and prepare each physical qubit in $\mathcal{G}_\textrm{anc}$ in the $\ket{0}$ state. We then measure all the stabilizer generators in $\mathcal{G}_\textrm{merged}$ and perform a round of error correction. To ensure this procedure is fault tolerant in the presence of noisy measurements we can repeat this step $d$ times~\cite{Vuillot19}. Once we have fault tolerantly obtained the result of the measurement of $\tilde{X}$ we can return to the original code space $\mathcal{C}$ by measuring each physical qubit in $\mathcal{G}_\textrm{anc}$ in the Pauli-$Z$ basis. 

Lemma~\ref{lemma:merge} provides a mechanism to perform a measurement of an $X$ logical operator $\tilde{X}$; however, it is restricted to the special case where there is no other $X$ logical operator within its support.  In the general case, we may have another $X$ logical operator $\tilde{X}'$ supported entirely on a strict subset $V'$ of physical qubits in $V[1]$. Following the construction as given above, we will make an unwanted measurement of $\tilde{X}'$, which will result in an entirely different computation.  As an example of this situation, consider the measurement of a two-logical-qubit operator $\tilde{X}_1 \tilde{X}_2$, where $\tilde{X}_1$ and $\tilde{X}_2$ are canonical logical operators that do not intersect at any physical qubits.  Following the procedure outlined above will give us separate measurements of $\tilde{X}_1$ and $\tilde{X}_2$.

We now generalize our construction of $\mathcal{G}_{\textrm{merged}}$ to address this general situation, which we illustrate in Fig.~\ref{fig:xx}.
To measure $\tilde{X}_1 \tilde{X}_2$ without measuring the value of $\tilde{X}_1$ or $  \tilde{X}_2$ individually, we must construct $\mathcal{G}_\textrm{merged}$ such that the separate logical operators are connected. First construct the separate  ancilla systems, $\mathcal{G}_{\text{anc}. 1}$ and $\mathcal{G}_{\text{anc}., 2}$ for $\tilde{X}_1$ and $\tilde{X}_2$. Let $c_1[k] \in C^T_1[k]$ and $c_2[k] \in C^T_2[k]$ be two arbitrary $X$ generators for $1 \leq k \leq r$. We introduce a new physical qubit $a[k]$ for $1 \leq k \leq r$ and a new $Z$ generator $z[j]$ for $1 \leq j \leq r$. There are the variable and check nodes in the highlighted region of Fig.~\ref{fig:xx}. The nodes of $\mathcal{G}_\textrm{merged}$ will consist of all the nodes in $\mathcal{G} \cup \mathcal{G}_{\text{anc}., 1} \cup \mathcal{G}_{\text{anc}., 2}$ as well as the nodes $a[k]$ and $z[j]$. The edge set of $\mathcal{G}_\textrm{merged}$ will also contain all the edges in $\mathcal{G} \cup \mathcal{G}_{\text{anc}., 1} \cup \mathcal{G}_{\text{anc}., 2}$ as well as the edges connecting $\mathcal{G}^T_{\tilde{X}_1}[2]$ and $\mathcal{G}^T_{\tilde{X}_2}[2]$ to $\mathcal{G}_{\tilde{X}_1}[1]$ and $\mathcal{G}_{\tilde{X}_2}[1]$ respectively. We then add the following edges, which are the edges in the highlighted region of Fig.~\ref{fig:xx}: $(a[j-1], z[j])$ and $(z[j], a[j])$ for $2 \leq j \leq r$, $(c_1[k], a[k])$ and $(a[k], c_2[k])$ for $1 \leq k \leq r$, and $(c_1^T[j], z[j])$ and $(z[j], c^T_2[j])$ for $2 \leq j \leq r$. Then the product of all $X$ checks in $\mathcal{G}\backslash \mathcal{G}_{\textrm{merged}}$ gives the measurement of $\tilde{X}_1 \tilde{X}_2$ and the product of the $X$ checks exclusively in $\mathcal{G}_{\textrm{anc}, 1}$ or $\mathcal{G}_{\textrm{anc}, 2}$ do not give measurements of $\tilde{X}_1$ or $\tilde{X}_2$ since they will have support on the physical qubits $a[k]$.

Before continuing, we offer some orienting remarks. First of all, if we choose to measure a logical operator supported at the boundary of the planar code defined using the lattice geometry presented in Ref.~\cite{Dennis02}, we recover the lattice surgery construction given in the original work~\cite{Horsman12}. 
\rev{Indeed, the gates we obtain are similar in spirit to those proposed in Ref.~\cite{Breuckmann17} where a surface code embedded on a torus is used as a resource to measure logical CSS operators of constant-rate hyperbolic surface codes. We expect the additional resources that will be needed to perform these logical operations will scale similarly with code distance to our scheme, up to constant factors. Our techniques however do not require the logical operators to have a specific structure, and are thus more broadly applicable to quantum LDPC codes.}
Secondly, the reader familiar with hypergraph product codes~\cite{Tillich14} can check that $\mathcal{G}_{\textrm{anc}}$ is the hypergraph product of $\mathcal{G}_{\tilde{X}}$ and the Tanner graph of a repetition code with $r$ variables. From this observation, it is easy to verify that the stabilizers supported on the variables of $G_\textrm{anc}$ commute.

\subsubsection{\label{sec:singley}Logical non-CSS measurement}

In order to implement the entire logical Clifford group using logical Pauli measurements, we must also be able to measure non-CSS logical operators that are the product of Pauli-X and Pauli-Z measurements~\cite{Brown17,Litinski19,chamberland2021universal}.  We now demonstrate how our construction can be adapted to such measurements.

First, we demonstrate how to measure the logical operator $\tilde{Y}$, which we illustrate in Fig. 4. Let $\tilde{X}$ be a logical operator in $\mathcal{C}$ and let $\tilde{Z}$ be the corresponding $Z$ logical operator. The corresponding $Y$ logical operator is given by $\tilde{Y} = i\tilde{X}\tilde{Z}$. Since $\tilde{X}$ and $\tilde{Z}$ must anti-commute any support of $\tilde{X}$ and any support of $\tilde{Z}$ must intersect at an odd number of qubits.

To measure $\tilde{Y}$ we prepare an ancilla system by combining two ancilla systems $\mathcal{G}_{\textrm{anc }\tilde{X}}$ and $\mathcal{G}_{\textrm{anc }\tilde{Z}}$ used to measure the CSS measurements $\tilde{X}$ and $\tilde{Z}$, respectively, using a fusion procedure as follows. We denote $A = V_{\tilde{X}}[1] \cap V_{\tilde{Z}}[1]$ as the set of qubits in the intersection of $\tilde{X}$ and $\tilde{Z}$. Their check operators are denoted $C_{\tilde{X}}[k] $ and $ C_{\tilde{Z}}[k]$, respectively.
 
For each physical qubit $v \in A$ there are corresponding $X$ generators $v^T_{\tilde{X}}[k] \in C^T_{\tilde{X}}[k]$ and $Z$ generators $v^T_{\tilde{Z}}[k] \in C^T_{\tilde{Z}}[k]$ for $1 \leq k \leq r$. Furthermore, each $v[1]$ has corresponding qubits $v_{\tilde{X}}[j]$ and $v_{\tilde{Z}}[j]$, for $2 \leq j \leq r$, in $\mathcal{G}_{\textrm{anc }\tilde{X}}$ and $\mathcal{G}_{\textrm{anc }\tilde{Z}}$ respectively. We then form $\mathcal{G}_\textrm{merged}$ using all the nodes and edges in $\mathcal{G}_{\textrm{anc }\tilde{X}}$ and $\mathcal{G}_{\textrm{anc }\tilde{Z}}$, except for the generators $v^T_{\tilde{X}}[k]$ and $v^T_{\tilde{Z}}[k]$ for all vertices in $ A$, which we replace with $y[k] = iv^T_{\tilde{X}}[k]v^T_{\tilde{Z}}[k]$. Merging the generators $v^T_{\tilde{X}}[k]$ and $v^T_{\tilde{Z}}[k]$ creates an extra degree of freedom which we fix by introducing check nodes $g[j]$ for $2 \leq j \leq r$ and adding a $Z$ edge $(v_{\tilde{X}}[j], g[j])$ and an $X$ edge $(v_{\tilde{Z}}[j], g[j])$. The product of the $X$ generators in $\mathcal{G}_{\textrm{anc }\tilde{X}}$, the $Z$ generators in $\mathcal{G}_{\textrm{anc }\tilde{Z}}$, and the generators $y[k]$ allow us to infer the measurement of $\tilde{Y}$. To check that these stabilizer generators commute note that all the generators in $\mathcal{G}_{\textrm{anc } \tilde{X}}$ and $\mathcal{G}_{\textrm{anc } \tilde{Z}}$ commute since they act on different qubits, except $v^T_{\tilde{X}}[1]$ and $v^T_{\tilde{Z}}[1]$, which we rectified by combining them. This procedure can also be straightforwardly adapted to measure a logical operator of the form $\tilde{X}_1 \tilde{Z}_2$ when the intersection of the supports of $\tilde{X}_1$ and $\tilde{Z}_2$ is not empty.

\begin{figure}
    \centering
    \includegraphics[width=0.45\textwidth]{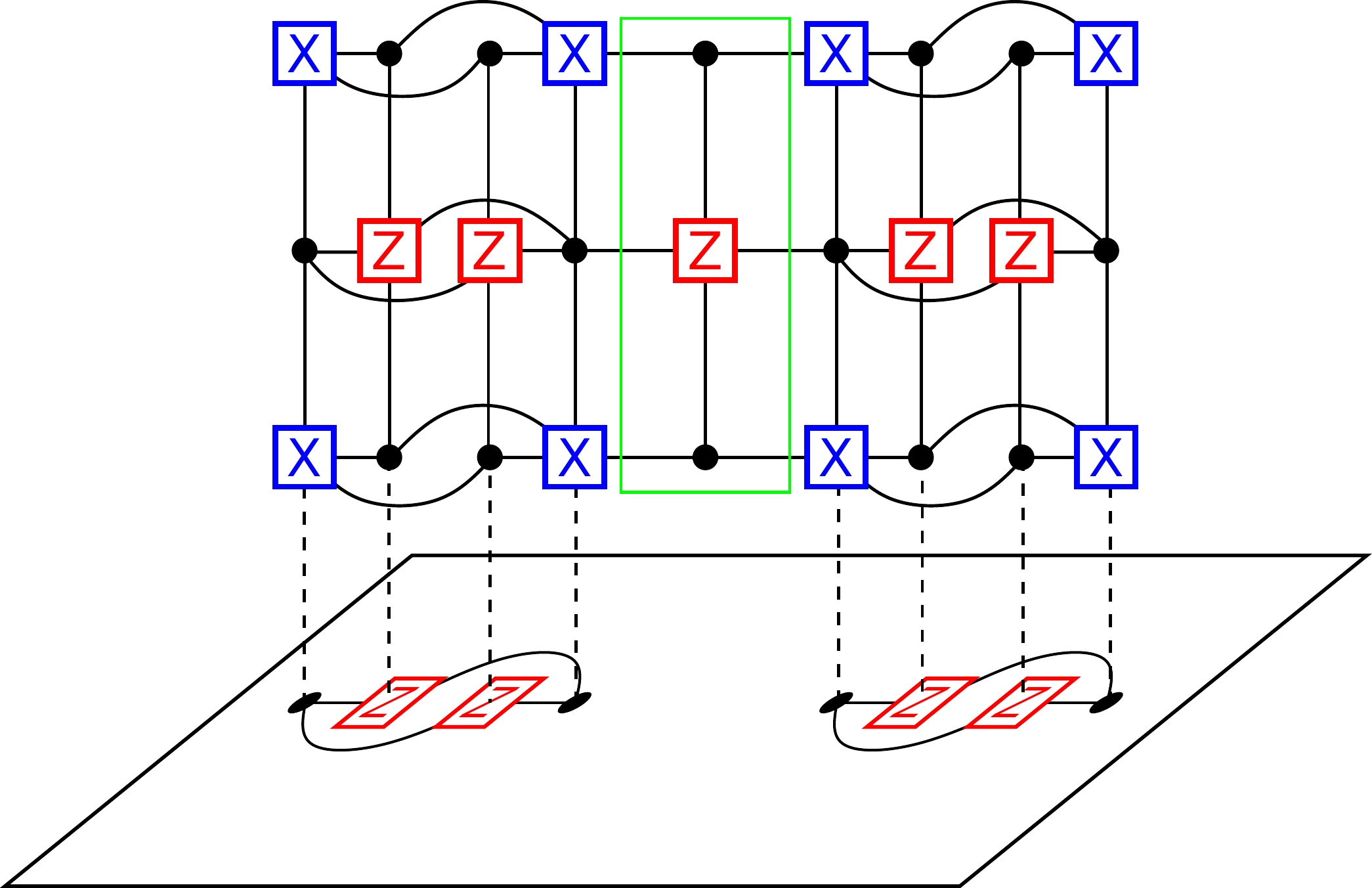}
    \caption{\textbf{Measurement of the logical operator $\tilde{X}_1 \tilde{X}_2$.} First, ancilla systems for the logical operators $\tilde{X}_1$ and $\tilde{X}_2$ are constructed as in Fig. 2. These ancilla systems are connected together as highlighted (green box). This ancilla system is then connected to the logical $\tilde{X}_1 \tilde{X}_2$ as previously. Observe that the product of the $X$ stabilizer generators in the ancilla system gives the logical $X_1X_2$. Furthermore, if we take the product of the $X$ stabilizers on the left ancilla system we do not obtain $X_1$, since this product will include qubits in the highlighted (green) region. The same holds for the stabilizers in the right ancilla system. Hence the product of the measurement results for these generators gives the measurement result for $X_1X_2$. After obtaining the measurement result, we can again measure the original stabilizers to return to the code space.}
    \label{fig:xx}
\end{figure}

Finally, we require measurement of non-CSS product operators such as $\tilde{X}_1 \tilde{Z}_2$ when $\tilde{X}_1$ and $\tilde{Z}_2$ do not intersect at any physical qubits, as illustrated in Fig. 5. As in the general case for CSS measurements, we must first connect the logical operators in order to make a parity measurement. We calculate the ancilla systems $\mathcal{G}_{\textrm{anc } 1}$ and $\mathcal{G}_{\textrm{anc } 2}$ for $\tilde{X}_1$ and $\tilde{Z}_2$. Let $c_1[k] \in C^T_1[k]$ and $c_2[k] \in C^T_2[k]$ be $X$ and $Z$ generators respectively for $1 \leq k \leq r$. As we did for the measurement of $\tilde{X}_1 \tilde{X}_2$, we extend out each ancilla system so that $\mathcal{G}_{\textrm{anc }1}$ has a boundary of $X$ generators coming out of the column defined by $c_1[k]$ and $\mathcal{G}_{\textrm{anc }2}$ has a boundary of $X$ generators coming out of the column defined by $c_2[k]$. As for the $Y$ measurement we then merge these boundaries by merging the corresponding $X$ and $Z$ generators at these boundaries and creating weight-two $XZ$ generators between corresponding physical qubits at these boundaries. Then if we take the product of the $X$ generators in $\mathcal{G}_{\textrm{anc }1}$, the $Z$ generators in $\mathcal{G}_{\textrm{anc }2}$, and the merged stabilizers we obtain a measurement of $\tilde{X}_1 \tilde{Z}_2$ without measuring $\tilde{X}_1$ and $\tilde{Z}_2$ separately.

\begin{figure}
    \centering
    \includegraphics[width=0.45\textwidth]{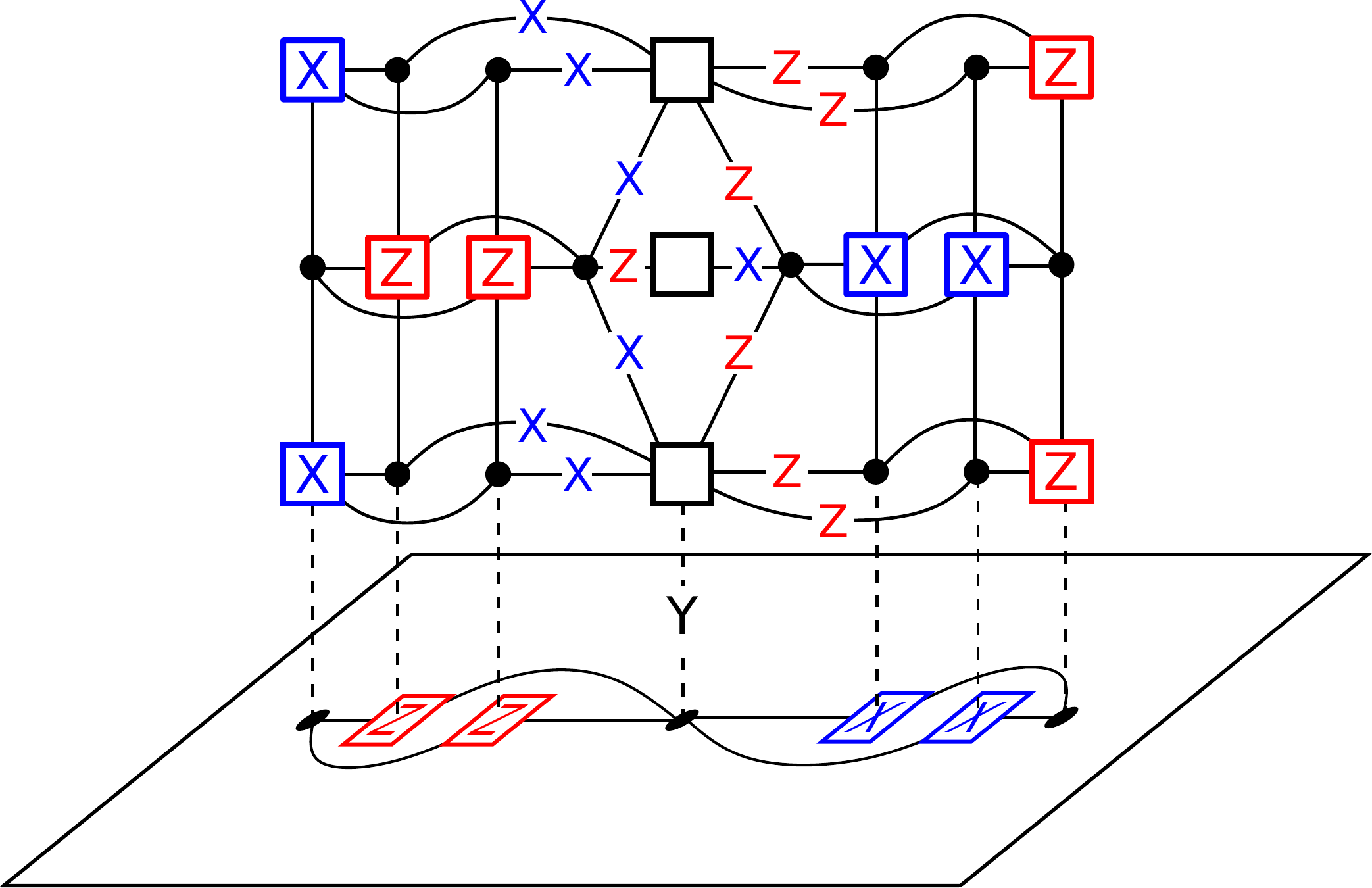}
    \caption{\textbf{Measurement of a logical $\tilde{Y} = i\tilde{X}\tilde{Z}$ operator.} This measurement closely follows that of logical $\tilde{X}_1\tilde{X}_2$ shown in Fig. 3, with the key difference being that ancilla systems for the logical operators $X_1$ and $Z_2$ are connected using non-CSS generators. Observe that the product of the $X$ stabilizer generators on the left, the $Z$ stabilizer generators on the right, and the mixed stabilizer generators in the dual layers gives the logical $\tilde{Y}$.}
    \label{fig:y}
\end{figure}

\begin{figure}
    \centering
    \includegraphics[width=0.45\textwidth]{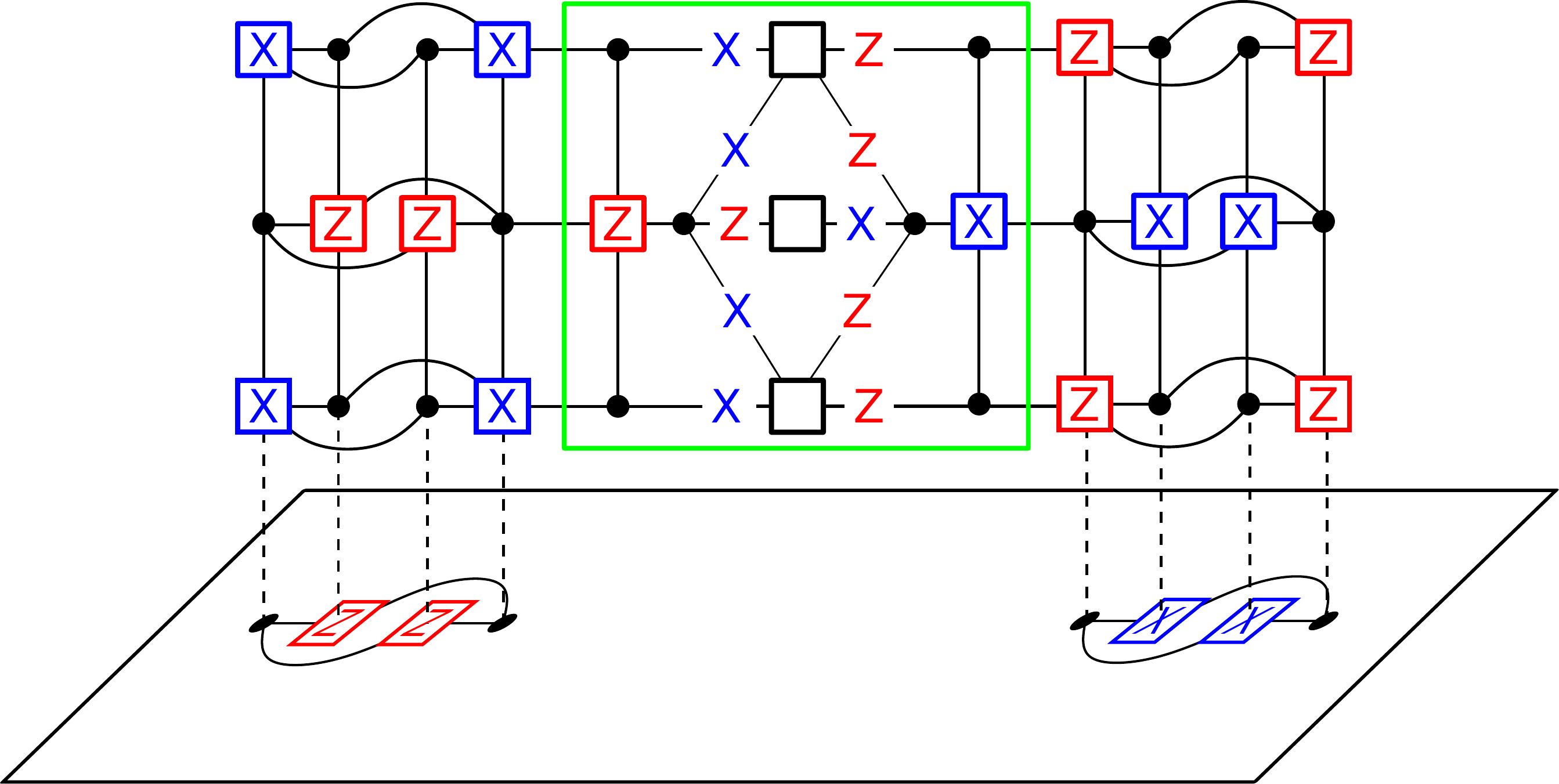}
    \caption{\textbf{Measurement of the logical operator $\tilde{X}_1\tilde{Z}_2$.} This measurement closely follows that of logical $\tilde{X}_1\tilde{X}_2$ shown in Fig. 3, with the key difference being that ancilla systems for the logical operators $X_1$ and $Z_2$ are connected using non-CSS generators (highlighted green box). Observe that the product of the $X$ stabilizer generators on the left, the $Z$ stabilizer generators on the right, and the mixed stabilizer generators in the highlighted region gives the logical $\tilde{X}_1\tilde{Z}_2$.}
    \label{fig:xz}
\end{figure}

\subsubsection{\label{sec:sim}Simultaneous measurement of commuting logical operators}

Suppose we wish to simultaneously measure two commuting logical operators. If the supports of the logical operators do not intersect then it is easy to see that we can make each measurement independently at the same time. If the supports do intersect but both logical operators are of the same type (all $X$ or all $Z$) then we can still independently measure each logical operator using our construction.  Finally, if the two commuting logical operators are not of the same type (say a Pauli-X type measurement $\tilde{X}_1$ and a Pauli-Z type measurement $\tilde{Z}_2$) and intersect then they must intersect at an even number of qubits. Let $\tilde{L}_1$ and $\tilde{L}_2$ be the two commuting Pauli operators and let $a$ and $b$ be two physical qubits in the intersection $V_1[1] \cap V_2[1]$ such that $\tilde{L}_1$ and $\tilde{L}_2$ act with different Pauli operators on $a$, and similarly on $b$. Then the generators $a^T[1] \in C_1^T[1]$ and $a'^T[1] \in C_2^T[1]$ will not commute. Similarly, the generators $b^T[1] \in C_1^T[1]$ and $b'^T[1] \in C_2^T[1]$ will not commute. To rectify this situation, we replace the generators $a^T[k]$ and $b^T[k]$ with $a^T[k]b^T[k]$ for $1 \leq k \leq r$ and create weight-two $Z$ stabilizer generators acting on $a[j]$ and $b[j]$, for $2 \leq j \leq r$ to fix the degrees of freedom created from merging the stabilizer generators. Similarly, we replace we replace the generators $a'^T[k]$ and $b'^T[k]$ with $a'^T[k]b'^T[k]$ and create weight-two $X$ stabilizer generators acting on $a'[j]$ and $b'[j]$. Since there are always an even number of intersections, non-commuting stabilizers can always be paired up.

\subsection{\label{sec:fault_tol}Fault tolerance}

Having presented our construction to perform logical Clifford gates via logical Pauli measurement, we now show that the merged code still possesses error correcting capabilities and therefore retains its fault-tolerance. First of all, it follows straightforwardly from our construction that the LDPC nature of the code remains intact throughout the deformation from $\mathcal{C}$ to $\mathcal{C}_{\text{merged}}$.

\begin{lemma}
Let $\tilde{L}$ be a logical operator on a CSS quantum LDPC code $\mathcal{C}$. Let $w$ be the maximum weight of a stabilizer generator in $\mathcal{C}$ and let $q$ be the maximum number of stabilizer generators acting on a physical qubit in $\mathcal{C}$. Let $w'$ and $q'$ be the equivalent quantities for the code $\mathcal{C}_{\text{merged}}$ obtained after making a measurement of $\tilde{L}$ in $\mathcal{C}$. Then $w', q' \leq \textrm{max}(w + 3, q+3)$. 
\end{lemma}

This fact is important for reducing complexity of the stabilizer measurement and for limiting the spread of errors throughout the deformation procedure.

\subsubsection{\label{sec:distance}Distance of $\mathcal{C}_{\text{merged}}$}

In order to preserve error correcting capabilities it is important that our code maintains a non-trivial distance throughout the code deformation. We show that the code distance of $\mathcal{C}_{\text{merged}}$ is no less than that of $\mathcal{C}$. There are two main concerns regarding the distance of our construction.  The first is whether the construction can reduce the weight of the logical operators. \rev{Through consideration of the structure of logical operators on the ancilla system}, we demonstrate that logical operators terminating at the boundary layers of the ancilla systems must necessarily maintain their weight because these boundaries are kept well-separated. \rev{Secondly, our construction can add new logical degrees of freedom to the code. As we are not interested in the state of these additional qubits, we refer to them as gauge qubits. This is consistent with terminology used in the context of subsystem codes, see Ref.~\cite{Poulin2005} for an introduction. We demonstrate that the logical operators associated to the gauge qubits (that we refer to as gauge operators) do not decrease the distance of the logical operators of interest. We call an operator that acts simultaneously on logical qubits and gauge qubits as a dressed logical operator. It is important to check that there are no dressed logical operators with weight smaller than $d$, as the existence of such an operator will mean that the code distance is decreased.}

\rev{Our following arguments will also use the notion of `cleaning'~\cite{Bravyi_2009}. We say that a logical operator $\tilde{L}$ is cleaned from some set of physical qubits $A$ to some set of physical qubits $B$, if we can multiply $\tilde{L}$ by an element $S$ of the stabilizer group such that the equivalent logical operator $\tilde{L}' = S \tilde{L} $  has trivial support on $A$ and non-trivial support on $B$.}

The distance of $\mathcal{C}_{\text{merged}}$ will depend upon the value of $r$, where $r$ was defined as the number of copies of $\mathcal{G}_{\tilde{L}}^T$ in the ancilla system. Let $d$ be the distance of $\mathcal{C}$. Consider the example of measuring the logical operator $X_1X_2$ of the surface code that encodes two logical qubits shown in Fig. 6. Then letting $r=1$ will significantly reduce the distance of the code, whereas letting $r=d$ will preserve the distance. In general we will show that if we let $r = d$, where $d$ is the distance of $\mathcal{C}$, then we can guarantee that the distance of the code during deformation does not drop below $d$.

\begin{figure}
    \centering
    \includegraphics[width=0.45\textwidth]{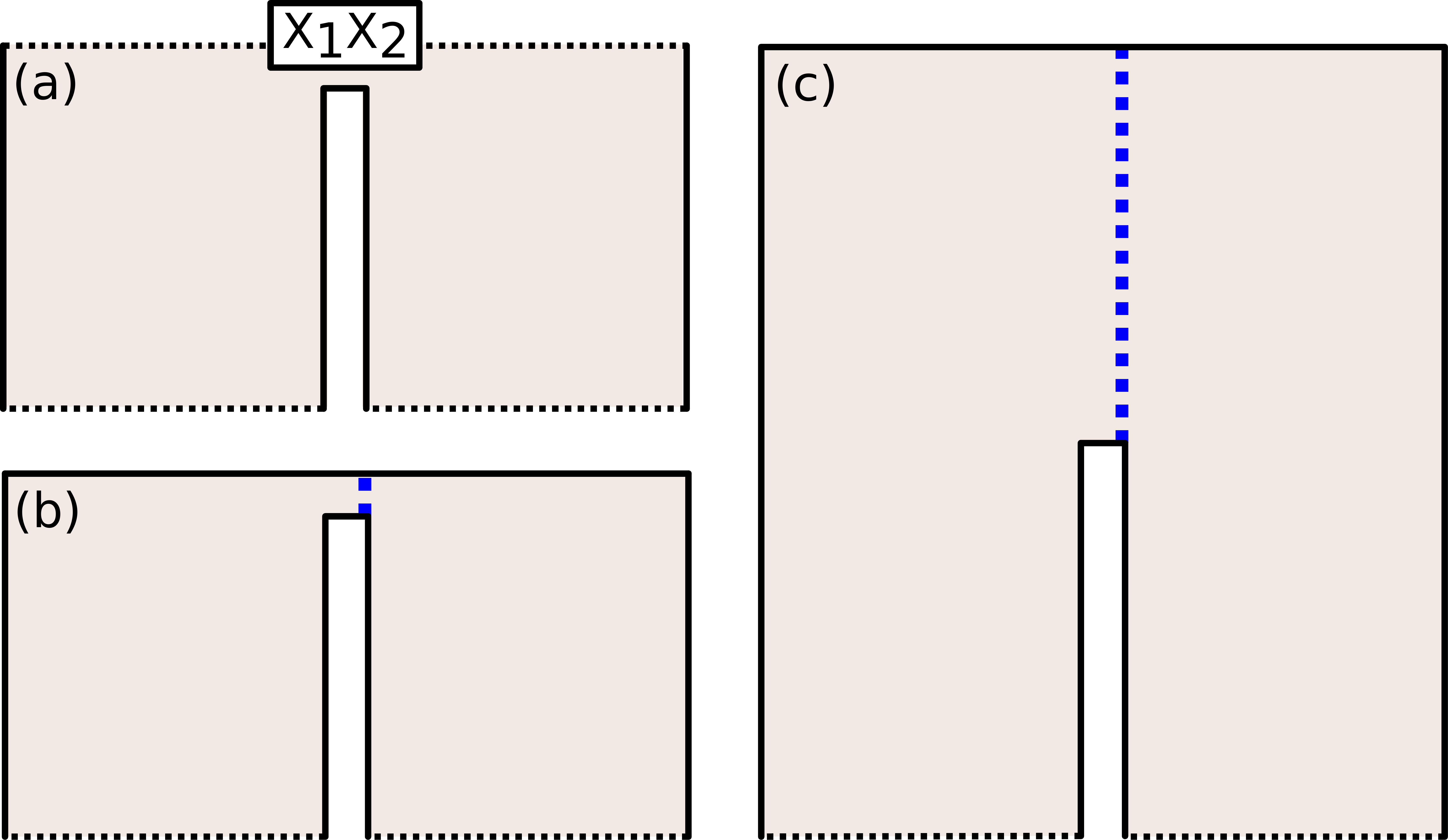}
    \caption{\textbf{Distance of $\mathcal{C}_{\rm merged}$.} \textbf{(a)} Suppose we have a code created by adjoining two surface code patches at the corner and we wish to measure the logical operator $X_1X_2$ along the top boundary. In \textbf{(b)} the ancilla system used for the measurement only has one layer, creating a low weight logical between the top and bottom smooth boundaries. To mitigate this we must use an ancilla with $d$ layers, as in \textbf{(c)}. This ensures that the distance of the code is preserved during the merge. It is worth noting that it may not always be the case that $d$ layers are needed to preserve the distance, and for certain codes it may suffice to only use one layer.}
    \label{fig:distance_example}
\end{figure}

Measuring a logical operator in $\mathcal{C}$ may introduce \rev{gauge} degrees of freedom for which we must account. Suppose we measure a logical operator $\tilde{L}$ in our code. We construct the ancilla system using the primal and dual graphs $\mathcal{G}_{\tilde{L}}[j]$ $\mathcal{G}^T_{\tilde{L}}[j]$ and merge them together with $\mathcal{C}$ to create $\mathcal{C}_{\text{merged}}$. If the number of independent check nodes in $\mathcal{G}^T_{\tilde{L}}$ is less than or equal to the number of variable nodes then $\mathcal{G}^T_{\tilde{L}}$ will have non-trivial degrees of freedom and making the measurement of $\tilde{L}$ will add extra \rev{gauge} \rev{qubits} to our code. \rev{In the case of an $\tilde{X}$ measurement, we can choose a canonical set of $Z$ gauge operators that are entirely contained in any dual layer of the ancilla system. Suppose that $r=1$ and $\mathcal{G}^T_{\tilde{L}}[1]$ contains $n'$ variable nodes and $m'$ check nodes. Then if we interpret $\mathcal{G}^T_{\tilde{L}}[1]$ as a classical code it contains at least $n'-m'$ logical bits. This is a lower bound since some of the checks in $\mathcal{G}^T_{\tilde{L}}[1]$ may be a linear combination of other checks. However in our case there are exactly $n'-m'+1$ logical bits in $\mathcal{G}^T_{\tilde{L}}[1]$. This is due to the fact that if there were a subset of check nodes in $\mathcal{G}^T_{\tilde{L}}[1]$ whose product gives the identity, then the equivalent qubits in $\tilde{L}$ would be a logical operator, and we have enforced the requirement that $\tilde{L}$ contains no subsets that support a logical operator. We will call the logical operators of $\mathcal{G}^T_{\tilde{L}}[1]$ cycle operators. However, we stress that the exact structure of these operators will not be important. Now when we create $\mathcal{G}_{\text{merged}}$ with $r=1$ we add $n'$ qubits, $m'$ stabilizer generators, and remove one logical qubit, thus adding at least $n'-m'+1$ new gauge qubits. However the same constraint for $\mathcal{G}^T_{\tilde{L}}[1]$ carries to $\mathcal{G}_{\text{merged}}$ and hence there are exactly $n'-m'+1$ new gauge qubits, and we can choose the canonical set of $Z$ gauge operators to be the logical operators of $\mathcal{G}^T_{\tilde{L}}[1]$.} These define all of the \rev{gauge} degrees of freedom.

The \rev{cycle operators} on $\mathcal{G}^T_{\tilde{L}}$ are independent of the original logical operators on $\mathcal{C}$, and correspond to stabilizers of $\mathcal{C}$. Each qubit in a cycle in $\mathcal{G}^T_{\tilde{L}}[1]$ has a corresponding stabilizer generator in $\mathcal{G}_{\tilde{L}}[1]$. By applying these generators we can \rev{clean} the \rev{cycle operator} so that it is supported entirely on $\mathcal{C}$. In this case we can see that the \rev{cycle operator} is thus equivalent to a product of old stabilizers in $\mathcal{G}_{\tilde{L}}[1]$. 

A cycle gauge operator in the bottom layer $\mathcal{G}^T_{\tilde{L}}[1]$ can be cleaned to other layers $\mathcal{G}^T_{\tilde{L}}[k]$ through the application of stabilizer generators in $\mathcal{G}_{\tilde{L}}[k]$ and so equivalent cycles in different layers of $\mathcal{G}^T_{\tilde{L}}[k]$ correspond to the same gauge qubit.

We illustrate this in Fig. 7, where we show the measurement of an $X$ logical operator. Here we can construct a \rev{canonical set of logical operators} of $\mathcal{G}^T_{\tilde{L}}[1]$ with two cycles and so there are two gauge qubits. In this case the cycles correspond to $Z$ gauge operators. The following lemma tells us about the nature of the $X$ gauge operators.

\begin{lemma} \label{lemma:gauge}
Let $\mathcal{C}_{\rm merged}$ be the code obtained after measuring an $\tilde{X}$ logical operator. Let $\tilde{Z}_g$ be a $Z$ gauge operator in $\mathcal{C}_{\textrm{merged}}$. If a logical or gauge operator $\tilde{O}$ anti-commutes with $\tilde{Z}_g$ then it has weight at least $r$.
\end{lemma}

\begin{proof}
If $\tilde{O}$ anti-commutes with $\tilde{Z}_g$ then it must intersect every possible support of $\tilde{Z}_g$. There are equivalent cycles all representing $\tilde{Z}_g$ on each dual layer of $\mathcal{G}_{\textrm{anc}}$. These equivalent cycles are disjoint and so the $X$ gauge operator must have support on every dual layer. Since there are $r$ dual layers, $\tilde{O}$ must have weight at least $r$, and furthermore must have weight at least one on each dual layer of $\mathcal{G}_{\textrm{anc}}$.
\end{proof}

We can thus think of the $\tilde{X}$ gauge operators as `strings' that travel from the top boundary $\mathcal{G}^T_{\tilde{L}}[r]$ to the bottom of $\mathcal{G}^T_{\tilde{L}}[1]$ and terminate in $\mathcal{C}$. We now show that $\mathcal{C}_{\text{merged}}$, when treated as a subsystem code, has distance at least $d$.

\begin{theorem}
Let $\mathcal{C} = \llbracket n, k, d \rrbracket$ be a quantum CSS LDPC code and let $\tilde{L}$ be a logical operator in $\mathcal{C}$. Let $\mathcal{C}_{\text{merged}}$ be the code obtained after making a measurement of $\tilde{L}$ in $\mathcal{C}$ using an ancilla system $\mathcal{C}_{\textrm{anc}}$. Then $\mathcal{C}_{\text{merged}}$, when treated as a subsystem code, has distance $\geq d$ as long as the ancilla system has at least $2d-1$ layers.
\end{theorem}

\begin{proof}
    See 'Materials and methods'
\end{proof}

\begin{figure}
    \centering
    \includegraphics[width=0.45\textwidth]{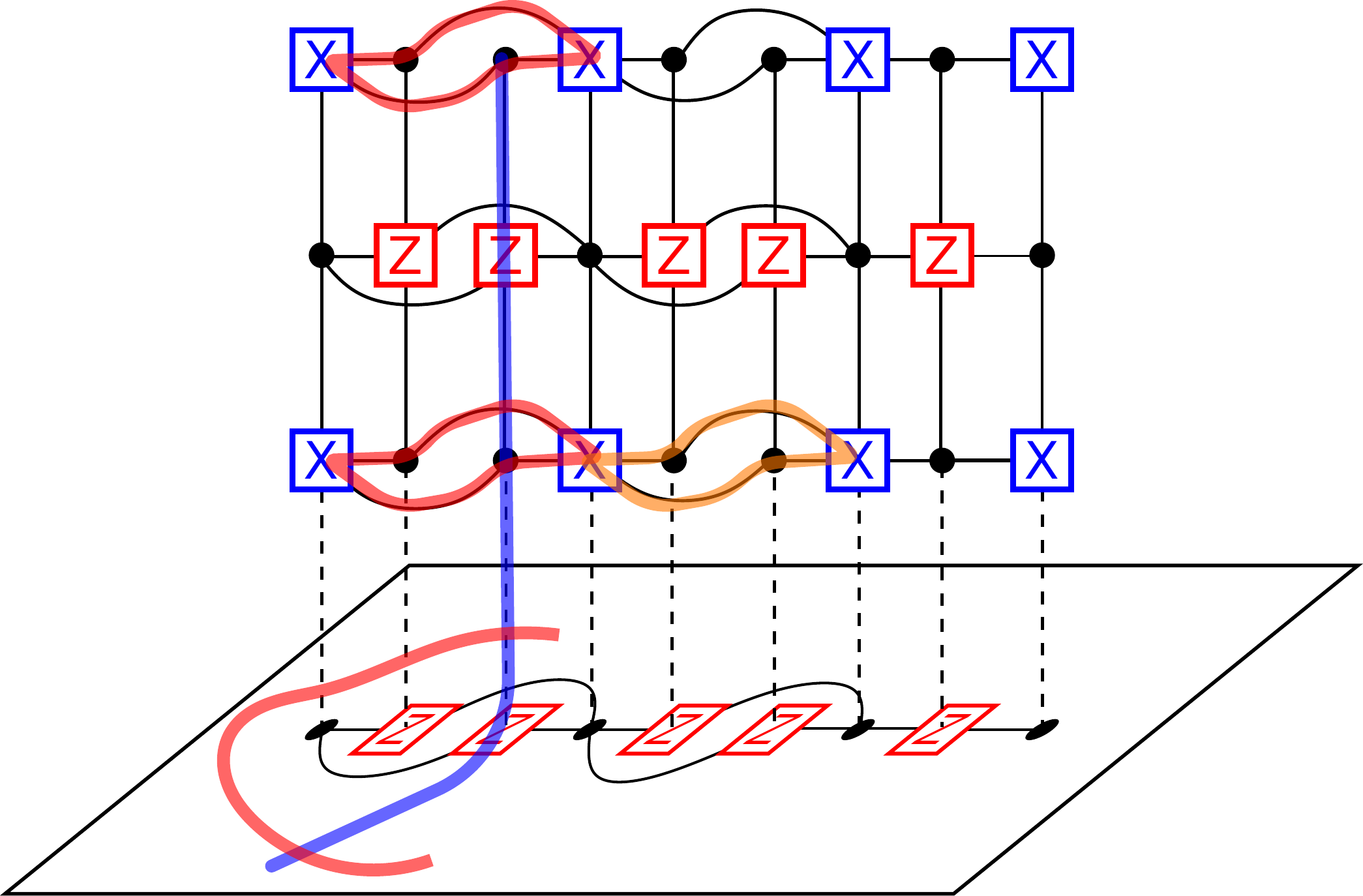}
    \caption{\textbf{Gauge operators of $\mathcal{C}_\textrm{merged}$.} The red and orange cycles in the first layer of the ancilla system give a canonical set for the $Z$ gauge operators. All red cycles correspond to the same gauge operator up to stabilizers. We can see that the $Z$ gauge operators can be \rev{deformed} so that they lie entirely in $\mathcal{C}$, and are equivalent to the product of the original $Z$ stabilizers that are in $\mathcal{G}_{\tilde{X}}$. The blue string is the $X$ gauge operator corresponding to the red cycles. It must anti-commute with any representation of the  conjugate $Z$ gauge operator and so it must intersect with every red gauge operator in the graph. Since the red gauge operator can be \rev{cleaned} to any dual layer of the ancilla system, the blue $X$ logical must have support at least on every dual layers of the ancilla system. Since we use an ancilla system with $d$ dual layers, this means an $X$ gauge operator must have weight at least $d$.}
    \label{fig:gauge}
\end{figure}

\subsection{\label{sec:FTQC}Low-overhead fault-tolerant quantum computation}

The previous section presents our construction for fault-tolerant measurements of logical Pauli operators.  These measurements allow for logical Clifford gates as well as initialization and measurement in the logical Pauli basis.  

Here, we consider what this construction means for performing low-overhead fault-tolerant quantum computing.  We first consider the space- and time-overheads associated with our approach to fault tolerance, with the \emph{parallelism} arising as a key quantity in determining these overheads.  We then discuss the overheads associated with non-Clifford gates implemented via magic states, and discuss the decoders available for quantum LDPC codes.

\subsubsection{\label{sec:space}Ancilla system size}

We now analyze the space overhead associated with making Pauli measurements using ancilla systems. Let $\tilde{L}$ be a logical of weight $w_L$ in a code $\mathcal{C}$ with distance $d$. Again, let $q$ be the maximum number of stabilizer generators connected to a physical qubit in $\mathcal{C}$. Then there are at most $qw_L/2$ check nodes in the graph $\mathcal{G}_{\tilde{L}}$ and to leading order the number of physical qubits in the ancilla system used to measure $\tilde{L}$ is

\begin{equation} \label{eq:size}
    \rev{n_a = (1 + q/2) w_L r} \,.
\end{equation}

\rev{In particular when $r = d$, as in Theorem 1, this means that the ancilla system is proportional to $w_L d$. This makes our scheme particularly applicable to codes such as the hypergraph product code, with distance scaling of $O(\sqrt{n})$.} Note that for the codes considered in Table 1, which are constructed from cyclic classical codes, we can choose a \rev{canonical set of logical operators} so that $w_L = d$ for all of the logical operators in this set. 

\subsubsection{\label{sec:parallel}Parallelism}

Let us now look at how the space and time overheads of our fault-tolerant scheme depend upon the weight of the logical parity measurements. We will see that there is a trade off between the space and the time overheads and that certain codes allow us to improve these overheads when compared to a general LDPC code. 

First note that the techniques for making parity measurements between two logical qubits can easily be generalised to make parity measurements between multiple logical qubits. For instance, suppose we wish to measure $\tilde{X}_1\tilde{X}_2\tilde{Z}_3$. This can be accomplished by first constructing the ancilla systems to measure $\tilde{X}_1\tilde{X}_2$ and $\tilde{Z}_3$ separately. We can then connect these ancilla systems to create the ancilla system to measure $\tilde{X}_1\tilde{X}_2\tilde{Z}_3$. We can then continue this construction to implement Pauli measurements of an arbitrary number of logical qubits. It is important that when we implement this construction we do not connect all the ancilla systems at the same qubit, as otherwise there can potentially be a qubit that is in the support of many stabilizer generators, and $\mathcal{C}_{\text{merged}}$ will not be an LDPC code.

There is however a limit to how many logical qubits we can include in a parity measurement if we want the code rate to remain constant during the deformation. Suppose the code has $\Theta(n)$ logical qubits. Then if we want to make a parity measurement of all $X$ logical operators in the code we would require $\Omega(nd)$ ancilla qubits. If $d > \Omega(1)$ then the number of ancillas we must add is $> \Omega(n)$ and hence the resource cost of computation will diminish the savings made by choosing a finite rate code. From an asymptotic standpoint, this means that the number of logical qubits measured in a single parity measurement should be kept at a constant.  This restriction will limit the number of logic gates that can be performed in parallel, and consequently the time overhead of the computation is increased. In general there is a trade-off between space and time overhead. If we increase the space overhead by adding many ancilla systems then we can decrease the time overhead by allowing very large multi-logical measurements. 

We encapsulate these notions in the \emph{parallelism} of our fault-tolerant quantum computing scheme, which we defined in the summary of results. The parallelism is \emph{not} inherent to a code, and is instead chosen depending on the available space and time overhead. For instance, imagine a quantum computing platform for which the space overhead, i.e., required number of qubits, is the primary constraint. In this case, it will be advantageous to reduce the parallelism at the expense of time. In this space-constrained regime, the most natural choice is to use a parallelism of two, performing a sequence of weight-two logical Pauli measurements, which can be easily converted to an entangling gate.

However, certain codes, e.g., hypergraph product codes constructed from good classical LDPC codes, possess an extra structure that allows us to achieve the same parallelism with a more modest resource. We can lay out a hypergraph product code on a two-dimensional grid so that the support of a canonical logical operator is contained entirely in one row or one column~\cite{Kovalev2013}. Each row and column can support multiple logical qubits and so we can configure an ancilla system for an entire row or column that can measure multiple logical qubits. In the asymptotic limit, each row or column contains $\Theta(\sqrt{n})$ physical qubits and $O(\sqrt{n})$ logical operators and so for a $d = O(\sqrt{n})$ hypergraph product code a single ancilla system of size $O(n)$ can be to used to measure Pauli operators with logical weight $O(\sqrt{n})$. A naive scheme would have yielded an ancilla system of size $O(wd^2 ) = O(n^{3/2})$. As an explicit example, the $\llbracket 7938, 578, 16 \rrbracket$ code from Table 1 can be arranged on a two-dimensional grid such that each row or column contains $63$ physical qubits and $17$ logical operators of the same type. Each qubit is in the support of at most $5$ stabilizer generators and the code has distance $16$. Hence using Eq.~(\ref{eq:size}) we find that we can construct an ancilla system for an entire row or column, and measure up to $17$ logical operators, using around $3528$ physical qubits. Furthermore, in such a layout logical operators of different type only intersect if they correspond to the same logical qubit, and hence if we assume that each logical qubit is only acted upon non-trivially by one logical operator in each round of error correction then we can keep the stabilizer generator weight low by avoiding the construction for the simultaneous measurement of commuting logical operators.

If we wish to be able to measure all the logical operators on the hypergraph product code then we require $O(n^{3/2})$ ancilla qubits. This reduces the rate and distance of the scheme to $O(n^{2/3})$ and $O(n^{1/3})$ respectively. Note that these code parameters cannot be achieved by a strictly local code in two dimensions because such codes necessarily obey the constraint $kd^2 = O(n)$~\cite{Bravyi2010}.

\subsubsection{\label{sec:scheme}Magic states}

With a scheme to perform Pauli measurements on quantum LDPC codes, we consider how this approach can be integrated into a broader scheme for performing universal fault tolerant quantum computing with low overhead. There is one main requirement that we still need in addition to Pauli measurement, which is the distillation of magic states~\cite{Bravyi05}.

As we have already mentioned, Clifford gates alone are not sufficient to perform universal quantum computing, and in addition we require a non-Clifford gate such as a $T$ gate or a CCZ gate. A standard approach to fault-tolerant non-Clifford gates is through injection of magic states. For example, the $T$ gate can be implemented using the magic state $\ket{T}$ and Clifford gates. Unless the code being used has transversal $T$ gates it is generally difficult to prepare these magic states in a fault-tolerant way. For this reason, magic state distillation, which prepare a small number of low error magic states from a larger number of noisy magic states, is often required for fault tolerant architectures. 

Here we assume that all the data qubits are stored on one LDPC code block, such as a hypergraph product code, along with ancilla systems to make Pauli measurements on the data block. To distill magic states we use a separate magic-state factory, and then inject the distilled magic states into our data block using the ancilla systems. We will consider previously designed magic-state factories that use surface codes. As an example computation, suppose we want to perform $10^{10}$ $T$ gates, the number required in~\cite{Gidney21}, with a tolerance of $1\%$, and each noisy $\ket{T}$ has an error of $10^{-3}$. Then we require a distillation scheme with an output error rate of $10^{-12}$. Such a magic-state distillation scheme can be implemented using $\sim 15000$ physical qubits where we do not count the qubits needed for stabilizer measurements. Such a distillation scheme in conjunction with the $\llbracket 7938, 578, 16 \rrbracket$ data block would still render a favourable overhead when compared to a full surface-code scheme. Of course this would render a fairly slow scheme, since we would only be producing one magic state at a time. In the discussion we consider the possibility of further reducing the overhead required for magic-state distillation by using LDPC codes, allowing us to increase the frequency of magic-state production while maintaining low overhead.

\subsubsection{\label{sec:decoding} Decoders}
There has been extensive work in designing efficient algorithms for decoding quantum LDPC codes. 
Many of these adapt known algorithms for classical LDPC codes, modifying them to deal with the nuances of quantum codes. 
One simple decoding algorithm for classical LDPC codes is the bit-flip algorithm. Leverrier, Tillich, and Zemor~\cite{Leverrier15} adapted this decoder for quantum LDPC codes and designed the small-set flip decoder, which was then shown to be able to correct a linear number of errors on quantum LDPC codes with sufficient expansion properties~\cite{Fawzi20}. The predominant algorithm for classical LDPC codes is belief propagation (BP) decoding. BP works by envisioning the code as a graph and transmitting likelihoods between the nodes. While BP works well on classical codes its performance is not as consistent on quantum codes due the degeneracy present in quantum LDPC codes, which results in split \rev{beliefs, where the decoder is not able to choose between two equivalent corrections}. As a result several modifications have been proposed to adapt BP decoding for quantum LDPC codes. In particular, Pantaleev and Kalachev combined BP with ordered statistic decoding (OSD) to design a decoder that appears to perform well on a variety of quantum LDPC codes~\cite{Panteleev19}. Hastings~\cite{Hastings14} created an efficient greedy decoding algorithm that is able to correct a constant number of errors on the hyperbolic surface codes. There is thus a sufficient body of work showing that efficient decoding of quantum LDPC codes is possible. See also~\cite{Breuckmann16,Breuckmann17,Breuckmann20}.
For decoding we will assume the use of one of the decoders outlined above. Given a good choice of code, these decoders appear to have performance comparable or even superior to that of the surface code.

\section{\label{sec:discussion} Discussion}
We have shown that it is possible to use LDPC codes to achieve fault-tolerant quantum computing with overheads favourable to surface code schemes, even at reasonable scales. Our scheme uses a generalized form of lattice surgery, which when coupled with magic-state distillation can implement universal quantum computing.

\rev{Our construction highlights parallelism as a key constraint in maintaining low overheads, and we note that this role of parallelism has also been identified in other schemes for quantum computing using LDPC codes, in particular the scheme by Gottesman~\cite{Gottesman14}. In Ref.~\cite{Gottesman14}, quantum computation proceeds by preparing encoded resource states and then using teleportation to execute quantum gates.  A concatenation scheme is used to prepare arbitrary resource states into the LDPC code blocks, which has a non-negligible overhead, and thus requires a fixed value of the parallelism in order to maintain a constant overhead.  In constrast, our proposal allows for parallelism to increase with distance for certain code families, such as the hypergraph product codes. Furthermore, while the scheme in Ref.~\cite{Gottesman14} has the required asymptotic behavior, concatenation often suffers from poor error thresholds as well as undesirable overheads in the practical regimes of interest. For our proposal, based on code deformation, we can expect similar thresholds to those obtained for LDPC codes used as quantum memories~\cite{tremblay2021constantoverhead}.}

There are two main bottlenecks, in terms of overhead, in our scheme, and we now address how these may be overcome. One of the main contributions to the overhead in our scheme is magic-state distillation. Although we argued that it is possible to use surface-code schemes while still maintaining an overhead advantage, it is generally necessary to use several magic-state factories to achieve a high enough rate of magic-state production. There is a potential for LDPC codes to be used to improve the overhead of magic-state distillation schemes.  We can look to some ideas from surface codes. Reference~\cite{Litinski19} presents several surface-code distillation schemes that achieve low overheads using the following idea. Circuits for magic-state distillation using codes with transversal $T$ gates can be rewritten to use only $Z$ measurements. This means that $Z$ logical errors are far less destructive than $X$ logical errors, and hence it is sufficient to use surface codes with low $d_z$. It is very straightforward to generalize this to LDPC codes such as the hypergraph product code. By taking the product of a high rate, low distance classical LDPC code and a high distance repetition code, we can construct LDPC codes with low $d_z$, high $d_x$, and higher encoding rates than the surface code. While this will not improve the overhead for constructing a single magic state factory, it may allow us to fit several identical factories in the same space as one surface-code factory.

One advantage that such a magic-state scheme would offer is that the hypergraph product code used is only non local in one dimension. This may be a desirable simplification in some quantum systems. In general, a promising line of research is to consider how much we can gain from LDPC codes while restricting the level of non locality allowed. It has been shown that there is a connection between the distance of the code and the connectivity~\cite{Krishna21,delfosse2021bounds}, and that limiting the non locality will limit the attainable distances~\cite{Baspin21}. This may however be a sacrifice that is necessary to see quantum LDPC codes physically implemented.

The other main contribution to the overhead comes from the ancilla systems used to make logical measurements. As we established, in order to maintain the distance of our code during the code deformation, we require that the ancilla systems have a height of at least $d$. Lowering this overhead while maintaining the required distance may be achieved using a technique established in Ref.~\cite{Hastings21}. In this reference Hastings proposes increasing the expansion of the graph $\mathcal{G}_{\tilde{L}}$, and as a result the number of layers in the ancilla system would only need to be constant in size, as opposed to scaling with $d$. In order to implement this approach, we would require a deterministic method for creating graphs with sufficient expansion. While constructions for Ramanujan graphs with high expansion exist, further work is needed to see how these can be integrated into our scheme. We also note that it is unclear that this proposal will offer a significant overhead improvement at the scales that we have considered, and instead give improvements to the resource cost of quantum computing at larger scales. Furthermore, it will be worthwhile to investigate the time-overhead of fault-tolerant gates. The discovery of ancilla systems that enable code deformations by single-shot error correction~\cite{Bombin2015,Fawzi20,delfosse2020short} would permit measurement-based gates in constant time, leaving open the tantalising prospect of fault-tolerant quantum computing with constant-space \emph{and} constant-time overhead. 

\section{Materials and methods}

\begin{figure}
    \centering
    \includegraphics[width=0.45\textwidth]{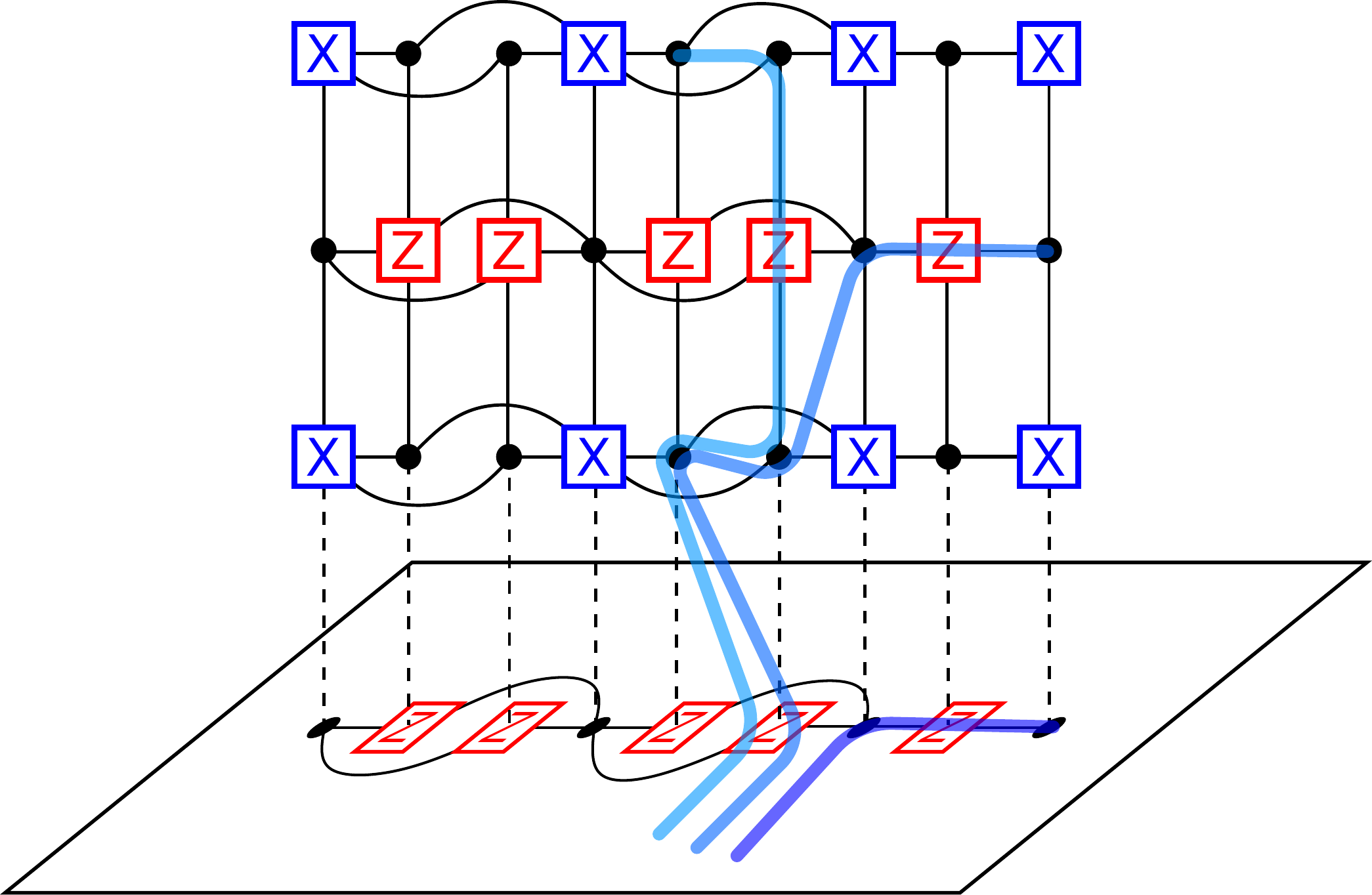}
    \caption{\rev{\textbf{Equivalent logical operators.}  The three curves here all represent equivalent $X$ logical operators. The dark curve is an $X$ logical operator on the original code $\mathcal{C}$. The other two logical operators are obtained by successive application of $X$ stabilizer generators in the dual layers of the ancilla system. In particular, after each application of $X$ stabilizer generators in a dual layer, the logical operator has support on the same dual layer since any subset of $X$ stabilizer generators in a dual layer are independent.}}
    \label{fig:proof}
\end{figure}

We now restate and prove Theorem~1.

\begin{theorem2}
    Let $\mathcal{C} = \llbracket n, k, d \rrbracket$ be a quantum CSS LDPC code and let $\tilde{L}$ be a logical operator in $\mathcal{C}$. Let $\mathcal{C}_{\text{merged}}$ be the code obtained after making a measurement of $\tilde{L}$ in $\mathcal{C}$ using an ancilla system $\mathcal{C}_{\textrm{anc}}$. Then $\mathcal{C}_{\text{merged}}$, when treated as a subsystem code, has distance $\geq d$ as long as the ancilla system has at least $2d-1$ layers.
\end{theorem2}
\begin{proof}

We prove the theorem for the simple case where $\tilde{L}$ is an $X$ logical operator. Let us also assume without loss of generality that $\tilde{L}$ does not contain any $X$ type stabilizers as a subset of $V_{\tilde{L}}[1]$. The proof of more general logical parity measurements can be obtained with the methods we develop here with the most simple case. 

Let $\tilde{Z}$ be a $Z$ logical operator contained entirely in $\mathcal{C}$. First note  that the $X$ stabilizers in $\mathcal{C}$ are left unchanged in $\mathcal{C}_{\textrm{merged}}$ and so $\tilde{Z}$ must still have weight at least $d$ on the physical qubits in $\mathcal{G}_\textrm{merged} \backslash \mathcal{G}_{\textrm{anc}}$. We now need to show that the application of a gauge operator and an arbitrary stabilizer does not reduce the weight of $\tilde{Z}$ below $d$. We are only interested in $Z$ type stabilizers and gauge operators as $X$ type stabilizers and gauge operators will not reduce the weight of a $Z$ logical operator. Let us consider what happens when we apply the $Z$ generators $C[1] \subset \mathcal{G}_{\tilde{L}}[1]$.
These generators are identical to the original generators from $\mathcal{C}$ with the addition of one extra physical qubit from $V^T[1] \subset \mathcal{G}^T_{\tilde{L}}[1]$ in each generator. Hence, after applying these generators \rev{$\tilde{Z}$} will still have weight at least $d$ on the physical qubits of $\mathcal{G}_\textrm{merged} \backslash \mathcal{G}_{\textrm{anc}}$. Stabilizer generators $C[k]$ for $k \ge 2$ can only increase the weight of the logical operator beyond $d$, as they have no common support with $\mathcal{G}_\textrm{merged} \backslash \mathcal{G}_{\textrm{anc}}$.

The $Z$ gauge operators when we measure an $X$ logical operator are cycles in $\mathcal{G}^T_{\tilde{Z}}[1]$. As discussed above, these cycles can be cleaned entirely into $\mathcal{C}$. These gauge operators are stabilizers of $\mathcal{C}$, and since the support of $\tilde{Z}$ in $\mathcal{C}$ remains the same in $\mathcal{C}_{\text{merged}}$, the application of a stabilizer of $\mathcal{C}$ will not reduce the weight of $\tilde{Z}$ below $d$ since the weight of a logical operator is lower bounded by $d$ even with arbitrary stabilizer generators applied.

Let us now check that we do not change the weight of any $X$ logical operators by adding the ancilla system. Let $\tilde{X} \not= \tilde{L}$ be an $X$ logical operator contained entirely in $\mathcal{C}$. We must show that the weight of $\tilde{X}$ remains above $d$ with arbitrary application of $X$-type stabilizers and gauge operators. Let $\tilde{X}_g$ be an arbitrary, non-trivial $X$ gauge operator and let $\tilde{Z}_g$ be any $Z$ gauge operator that anti-commutes with $\tilde{X}_g$. Then $\tilde{X}\tilde{X}_g$ anti-commutes with $Z_g$ and so by Lemma~\ref{lemma:gauge} $\tilde{X} \tilde{X}_g$ has weight at least $d$, even with application of arbitrary stabilizers. 

The last case to consider is when no gauge operators are applied and we apply arbitrary stabilizers. This is only relevant if $\tilde{X}$ intersects with $\tilde{L}$ at some physical qubits. In this case we can apply an $X$ generator in the dual layer $\mathcal{G}^T_{\tilde{L}}[1]$ which will \rev{clean $\tilde{L}$ from the qubits at the intersection of $\tilde{L}$ and $\tilde{X}$ onto the equivalent qubits in the first primal layer of the ancilla system $\mathcal{G}_{\tilde{L}}[2]$}.

Doing this will also create non-trivial support on the first dual layer $\mathcal{G}^T_{\tilde{L}}[1]$, \rev{see Fig. 8}. To see why this is the case suppose there is a subset of $X$ generators in $C^T_{\tilde{L}}[1]$ such that their product gives trivial support on $V^T_{\tilde{L}}[1]$. That is, there is a subset $A^T$ of $C^T_{\tilde{L}}[1]$ such that each qubit in $V^T_{\tilde{L}}[1]$ is connected to an even number of generators in $A^T$. This implies there is a subset $A$ of $V_{\tilde{L}}[1]$ such that each check node in $C_{\tilde{L}}[1]$ is connected to an even number qubits in $A$. This means that if we apply $X$ to each physical qubit in $A$ to form the operator $\tilde{X}_A$ then $\tilde{X}_a$ commutes with each $Z$ generator in $C_{\tilde{L}}[1]$ and hence $\tilde{X}_A$ is an $X$ logical operator in $\mathcal{C}$, but we have assumed that $\tilde{L}$ contains no logical operators as a subset and hence there is no such set $A^T$ that gives trivial support on $V^T_{\tilde{L}}[1]$. Note in the general case we can have subsets of $V_{\tilde{L}}[1]$ that support a logical operator, however applying all the equivalent stabilizers in the first dual layer of $\mathcal{G}_{\textrm{anc}}$ will create support on the qubit connecting the two ancilla systems. 

The above argument tells us that \rev{we can clean $\tilde{L}$ from qubits in  $V_{\tilde{L}}[1]$, replacing them with qubits in $V_{\tilde{L}}[2]$} while creating non-trivial support on $V_{\tilde{L}}^T[1]$. Now we can continue this process to \rev{clean $\tilde{L}$ from qubits in $V_{\tilde{L}}[2]$ to qubits in $V_{\tilde{L}}[3]$ in the support of $\tilde{L}$} while creating non-trivial support on $V^T[2]$, and so on. For each qubit in the intersection of $\tilde{X}$ and $\tilde{L}$ that we \rev{cleaned from $\tilde{L}$}, there will always be an equivalent qubit in the primal layer $\mathcal{G}_{\tilde{L}}[j]$, until we reach the top boundary, at which point there will be support on at least one physical qubit in $V_{\tilde{L}}[j]$, and so the weight of $\tilde{X}$ will always be at least $d$. 
\end{proof}

\begin{acknowledgements}
We gratefully acknowledge David Poulin for early discussions on using quantum LDPC codes for computation.  We thank Andre Saraiva, Michael Vasmer, and Paul Webster for discussions and comments on the manuscript.
This work is supported by the Australian Research Council via the Centre of Excellence in Engineered Quantum Systems (EQUS) project number CE170100009, and by the ARO under Grant Number: W911NF-21-1-0007.  Competing Interests: The authors declare that they have no competing interests.
Author contributions: All authors, LZC, IHK, SDB and BJB conceived of the methodology and derived the main results. All authors contributed to the writing of the manuscript. 
Data availability: All data needed to evaluate the conclusions of this work are present in the paper.
\end{acknowledgements}

\bibstyle{plain}
\bibliography{apssamp}

\end{document}